\documentclass[runningheads,orivec]{llncs}

\usepackage[T1]{fontenc}
\usepackage{graphicx}
\usepackage[bookmarks,bookmarksopen]{hyperref}
\usepackage{amsmath}
\usepackage{amssymb}
\usepackage{color}
\usepackage{paralist} 
\usepackage{stmaryrd} 

\usepackage{tikz}
\usetikzlibrary{arrows.meta, positioning}

\tikzset{
  ew/.style={
    to path={(\tikztostart.east) -- (\tikztotarget.west)}
  }
}

\usepackage{xspace}
\usepackage{local}
\usepackage{todonotes}

\begin{document}
\title{Polymorphism Meets DHOL}

\author{Rhea Ranalter\inst{1}\orcidID{0009-0006-2861-548X} \and
Cezary Kaliszyk\inst{2,1}\orcidID{0000-0002-8273-6059} \and
Florian Rabe\inst{3}\orcidID{0000-0003-3040-3655}
}

\authorrunning{R. Ranalter et al.}

\institute{University of Innsbruck, Computational Logic, Austria\\
\email{r2analter@gmail.com}
\and
University of Melbourne, School of Computing, Australia\\
\email{ckaliszyk@unimelb.edu.au}
\and
University of Erlangen-Nuremberg, Computer Science, Germany\\
\email{florian.rabe@fau.de}
}

\maketitle

\begin{abstract}
 DHOL is an extensional, classical logic that equips the well-known higher-order logic (HOL) with dependent types.
This allows for concise encodings of important domains like size-bounded data structures,
category theory, or proof theory.
Automation support is
obtained by translating DHOL to HOL, for which powerful modern automated theorem provers are available.
However, a critically missing feature of DHOL is polymorphism.
We develop the syntax and semantics of polymorphic DHOL and extend the translation
accordingly. We implement 
the translation
in the
logic-embedding tool and
evaluate it on a range of TPTP formalizations. The logic-embedding tool, together with an
off-the-shelf HOL theorem prover easily creates a PDHOL theorem prover for experimenting.

\end{abstract}

\section{Introduction}
\paragraph*{Setting}
Monomorphic dependently typed higher-order logic (DHOL) was introduced in~\cite{RRB:dhol:23,RRB23ext}.
It combines the simplicity of higher-order logic (HOL) \cite{churchtypes,hol}, particularly the
use of extensionality and classical Booleans, with the often-wished-for feature of dependent types.
Contrary to proof assistants based on dependent type theory \cite{coq,agda,lean}, it uses dependent types in the simplest possible setting and, in particular, does not introduce universes or inductive types.
Instead, it only changes HOL's simple function type $A\to B$ into a dependent one $\P x A B$ and allows for base types $a$ to depend on typed arguments.

This makes typing undecidable~\cite{H97} but has the benefit of being close
to languages in the HOL ecosystem for which strong automated theorem proving
(ATP) support exists. 
This is leveraged by giving a linear, compositional,
and judgment preserving/truth reflecting translation from monomorphic DHOL to monomorphic HOL to obtain implementations of
type-checking and theorem proving for DHOL \cite{RRB23ext}, based on partial equivalence relations (PERs) \cite{persemantics}.
Practical experiments show that this combination of expressivity and ATP support makes undecidable typing a price worth paying~\cite{dhol_satallax}.

Maybe surprisingly, even though the PER semantics of dependent types has been known for some time \cite{persemantics}, it is not widely known and the proofs are complex.
Indeed, as noticed during the preparation of \cite{southern_phd}, the literature contains examples of similar translations that were incorrectly assumed to be sound/complete.
Therefore, \cite{RRB23ext} focused on the simplest possible language and left out a number of common language extensions for HOL.

\paragraph*{Outline and Contribution}
We introduce PDHOL as an extension of DHOL with shallow (ML-style, rank 1) polymorphism and subtype definitions.
This greatly extends the expressivity of DHOL while retaining its simple definition and strong automation support.

To simplify the presentation, we split our exposition of PDHOL:
Section~\ref{sec:syn} and~\ref{sec:rul} present the syntax and proof system for basic polymorphism, and we relegate the presentation of two advanced features to Section~\ref{sec:dtv} and~\ref{sec:std}.
In Section~\ref{sec:trans}, we describe the translation from PDHOL to polymorphic HOL (PHOL) and show it is sound and complete.

PDHOL has been adopted as a TPTP standard in the form of the new DHF dialect~\cite{R+25}.
Using its TPTP representation, we have implemented our translation as a PDHOL$\to$PHOL logic-embedding.
This allows using any off-the-shelf PHOL theorem prover as a PDHOL theorem prover with negligible overhead.

We apply PDHOL in Section~\ref{sec:case} to give elegant formalizations of practically important problems, reaching a level that is typically only found in the rich languages of \emph{interactive} provers and not in simple logics like PDHOL that provide \emph{automated} proof support.
For example, Section~\ref{sec:case} shows how polymorphism and dependent types allow tracking key invariants in common polymorphic data structures, e.g., the dimensions and matrices or the black height of red-black trees.
In addition to plain type variables (e.g., $\Pg{\alpha:\type}\ldots$), PDHOL allows \emph{dependent} type variables, i.e., type variables that depend on terms (e.g. $\P\alpha{\nat\to\type}\ldots$).
We spell this out in Section~\ref{sec:dtv} and use it for formalizations that require abstracting over \emph{families} of types, such as heterogeneous lists.
Finally, in HOL, the core feature for building conservative extensions of theories is the subtype definition principle: a fresh type is introduced and axiomatized to be isomorphic to a subtype of an existing type.
In Section~\ref{sec:std}, we extend this feature to the dependently-typed case.
We then show how dependent types, polymorphism, and subtype definitions together allow formalizing complex concepts such as the type of homomorphisms between algebraic structures.

\section{Syntax}\label{sec:syn}
The grammar below shows both the DHOL language and our \hlp{extension} to PDHOL.
We also \hlm{mark} the parts that must be removed or adjusted to recover HOL as a fragment. 
\begin{center}
\begin{tabular}{l c l r}
  \(T\) & ::= & \(.\) | $T, a:\hlp{\Pt \valpha}\hlm{\P\vx\vA}\type$ | $T, c: \hlp{\Pt \valpha} A$ | $T, \ass\hlp{\Pt \valpha} \Phi$ & Theories \\
  \(\Gamma\) & ::= & \(\circ\) | \hlp{$\Gamma, \alpha:\type$} | $\Gamma, x:A$ | $\Gamma,\ass \Phi$ & Context\\
  \(\gamma\) & ::= & \(\bullet\) | \hlp{$\gamma, A$} | $\gamma, t$ | $\gamma,\checkmark$ & Substitutions\\
  \(A,B\) & ::= & $a\:\hlp{\vA}\: \hlm{\vt}$ |
                  $\hlp{\alpha}$ | $\hlm{\P xAB}$ | $\bool$ & Types\\ 
  \(t,u,\Phi\) & ::= & $c\:\hlp{\vA}$ | \(x\) | $\lam xAt$ | $t\:u$ | \(t\Rightarrow u\) | \(\eq Atu\) & Terms (incl. formulas $\Phi$)\\
\end{tabular}
\end{center}
We use the usual notations for the logical connectives and binders, which are all definable~\cite{andrews_truthproof}, and we write $A\to B$ as usual.
To avoid case distinctions, we will occasionally merge lists $\Pt \valpha\P\vx\vA$ of type and term variables into a single list $\Pg{\Delta}$ for a context $\Delta$.
Similarly, we may merge the list $\vA\ \vt$ of type and term arguments into a single substitution $\delta$.
We also use the following abbreviations for sequences of expressions:
\begin{center}
\begin{tabular}{|l|l|l|}
\hline
abbr. & expansion & remark \\
\hline
$\vA$ & $A_1\ \ldots\ A_n$ & types in a substitution, application or binding\\
$\vt$ & $t_1\ \ldots\ t_n$ & terms in a substitution or application\\
$\tp{\valpha}$ & $\tp{\alpha_1} \ldots \tp{\alpha_n}$ & type variables in a context or binding\\
$\vx:\vA$ & $x_1:A_1\ldots x_n:A_n$ & term variables in a context or binding\\
\hline
\end{tabular}
\end{center}

\begin{example}\label{ex:vectors}
We present fixed-length lists, sometimes called vectors, as a
running example. For that, we start with natural
numbers:
\begin{gather*}
  \nat :\type
  \qquad 0:\nat
  \qquad \suc:\nat\to\nat
  \qquad\+:\nat\to\nat\to\nat\\
  \ass\forall n:\nat.\+\ 0\ n =_{\nat} n
  \qquad \ass\forall n, m:\nat.\+\ (\suc\ n)\ m=_{\nat} \suc\ (\+\ n\ m)
\end{gather*}
Here, \(\nat\) is a non-dependent, or simple, base type for which a constant, a constructor, and
a function are declared. We will abbreviate \(\suc\ 0, \suc\ (\suc\ 0), ...\) with \(1,
2, ...\).

Everything stated so far is expressible in regular HOL --- neither dependent types nor
polymorphism play a role. We now extend this theory to vectors, keeping the highlighting
conventions used in the grammar.
\begin{gather*}
  \vect:\hlp{\Pt{\alpha}}\hlm{\P{n}{\nat}}\type
  \quad \nil:\hlp{\Pt{\alpha}}\vect\;\ha\;\hlm{0}\\
  \cons:\hlp{\Pt{\alpha}}\hlm{\P{n}{\nat}}\ha\to\vect\;\ha\;\hn\to\vect\;\ha\;(\suc\;\hn)\\
  \app:\hlp{\Pt{\alpha}}\hlm{\P{n,m}{\nat}}\vect\;\ha\;\hn\to\vect\;\ha\;\hm\to\vect\;\ha\;(\+\
  \hn\;\hm)
\end{gather*}
Removing the highlighted dependent part yields polymorphic, dynamic lists in PHOL while
instantiating the highlighted polymorphic part results in fixed-type vectors in DHOL. Doing
both of these gives fixed-type, dynamic lists in HOL. 

Note that, if one now wants to prove the associativity of \(++\), type checking the
statement would require a proof of the associativity of \(+\) --- turning type checking into an
undecidable problem in general.
\end{example}

\paragraph*{Contexts and Substitutions}
Contexts \(\Gamma\) are lists of local declarations, subject to $\alpha$-renaming and
substitution as usual: (i) type variables $\alpha:\type$ (ii) typed variables \(x:A\) (iii)
local assumptions $\ass \Phi$. 
\(\circ\) denotes the empty context.

Substitutions $\gamma:\Gamma\to\Gamma'$ provide type/term expressions for all type/term variables declared in $\Gamma$ in such a way that the assumptions made in $\Gamma$ are satisfied.
To track when $\Gamma$ contained an assumption that must be proved, we write $\checkmark$ in the corresponding place of a substitution.
If we extended the formulation with a language of proofs, the $\checkmark$ would be replaced
with an appropriate proof expression.
\(\bullet\) denotes the empty substitution.
We write $E[\gamma]$ for the result of substituting in expression $E$ according to $\gamma$, and
we abbreviate as $E[t]$ the common case where the substitution is the identity for all
variables but the last one.

\begin{example}\label{ex:term}
  Assume a typing expression \(E\) for a 2-element vector \texttt{[n, m]}, represented formally as
    \(\cons\ \nat\ 1\ n\ (\cons\ \nat\ 0\ m\ (\nil\ \nat)):\vect\ \nat\ 2\). 
  For this to be properly typed, the
  context must include typing statements for \(n\) and \(m\). For this example, we
  also add two assumptions resulting in the context: \(n:\nat, m:\nat, \ass n=_{\nat} 2,
  \ass m=_{\nat} 3\).
  
  A well-formed substitution \(\gamma\) for this context is \(2, suc\ 2,
  \checkmark, \checkmark\). \(E[\gamma]\) would then be \(\cons\ \nat\ 1\
  2\ (\cons\ \nat\ 0\ (\suc\ 2)\ (\nil\ \nat)):\vect\ \nat\ 2\). The \(\checkmark\)s in the
  substitution represent the proof obligations \(2=_{\nat} 2\) 
  and \(suc\ 2 =_{\nat} 3\) 
  which are trivial after unfolding the abbreviations used for numbers.
\end{example}

\paragraph*{Theories}
Theories $T$ are lists of global declarations: (i) base types $a$ depending on typed arguments
$x:A$ (ii) typed constants $c$ (iii) global assumptions (i.e. axioms) $\ass \Phi$.
$.$ denotes the empty theory.

Contrary to contexts, declarations in theories may depend on arguments, and this is the
mechanism how both monomorphic DHOL~\cite{RRB23ext} and our extension thereof are defined as generalizations of HOL.
The following table gives an overview of the possible combinations:
\begin{center}
\begin{tabular}{|l|l|l|l|}
\hline
 & \multicolumn{3}{c|}{declaration of} \\
depending on & type symbol & term symbol & global assumption \\
\hline
type variable & \multicolumn{3}{c|}{\hlp{allowed in PDHOL and PHOL}} \\
term variable & \hlm{allowed in DHOL} & definable via $\Pi$ & definable via $\forall$ \\
local assumption&  \multicolumn{2}{c|}{not allowed} & definable via $\impl$\\
\hline
\end{tabular}
\end{center}
DHOL arises from HOL by allowing type symbols to depend on term arguments.
PDHOL arises by allowing all declarations to depend on type arguments.
Three combinations are always definable and therefore redundant.
We exclude the remaining two cases: local assumptions as parameters of type/term symbols.
These would effectively allow declaring partial functions and partial type symbols.
We exclude those because, to our knowledge, they cannot be translated to HOL in a sound and complete way.

\paragraph*{Types and Terms}
Types $A,B,\ldots $ and terms $t,u,\ldots$ are formed from
\begin{itemize}
\item references to symbols declared in the theory, which
  must always be fully instantiated with type and term arguments where applicable
\item references to variables declared in the context
\item the usual production rules of the grammar for dependent function types $\P x AB$, function formation $\lam xAt$ and application $t\ u$,
\item production rules of the grammar for Booleans $\bool$ formed from typed equality $\eq A t u$ and dependent implication $\Phi\impl\Psi$.
\end{itemize}
Note that equality of \emph{terms} is a Boolean term.
Equality of \emph{types} is not---it will be a judgment in the type system.



\emph{Dependent} implication means that in $F\impl G$, the well-formedness derivation for $G$ may already assume $F$ is true.
This is critical for type-checking, e.g., in \(a =_A b \impl f\: a =_{B[a]} f\: b\) for a dependent function $f$.
We believe dependent implication is not definable from equality and therefore make it an additional primitive from which corresponding dependent variants of conjunction and disjunction can be defined.
While the loss of commutativity of conjunction/disjunction may seem odd at first, the behavior is well-known from short circuit
evaluation in some programming languages.

\section{Inference System}\label{sec:rul}
\begin{figure}[hbtp]
\begin{center}
	\begin{tabular}{|c|c|c|}
		\hline 
		Name & Judgment & Intuition \\ 
		\hline
		theories & $\Thy T$ & $T$ is a well-formed theory\\
		contexts & $\Ctx \Gamma$ & $\Gamma$ is a well-formed context\\
		substitutions & $\SubG \delta \Delta$ & $\delta$ is a well-formed substitution for $\Delta$\\
		types & $\TypeG A$ & $A$ is a well-formed type \\ 
		typing & $\TermG tA$ & $t$ is a well-formed term of well-formed type $A$ \\
		validity & $\TrueG \Phi$ & Boolean $\Phi$ is derivable \\ 
		equality of types & $\TEqG AB$ & well-formed types $A$ and $B$ are equal \\ 
		\hline 
	\end{tabular}
\caption{DHOL Judgments}\label{fig:judge}
\end{center}
\end{figure}

In Fig.~\ref{fig:judge}, we give the judgments of PDHOL.
These look the same as for DHOL---but as contexts $\Gamma$ can now contain
type variables, they correspond to polymorphic statements.
We use \(\leftarrow\) in the judgment for well-formed substitutions to avoid confusion with the
simple function type constructor \(\rightarrow\).
Fig.~\ref{fig:declrules} and~\ref{fig:exprules}
present the rules of PDHOL.
They arise as an 
extension of the rules of DHOL.
We use \(\in, \notin\) on lists in the obvious manner.
\newcommand{\rtext}[1]{\multicolumn{2}{l}{\text{#1}}}
\begin{figure}[htbp]
  \[\begin{array}{l}
    \text{Well-formed theories $T$}\\
    \rul[Rthyempty]{ThyBase}{}{\Thy .} \quad
    \rul[Rthytype]{ThyType}{\Thy T \quad a\notin T\quad \Ctx{\Delta}            \quad \Delta=\tp\valpha,\vx:\vA}{\Thy{T,a:\Pg\Delta\type}}\\[.3cm]
    \rul[Rthyterm]{ThyCon} {\Thy T \quad c\notin T\quad \Type{\Delta}{A}        \quad
    \Delta=\tp\valpha}        {\Thy{T,c:\Pt\Delta A}} \\ 
    \rul[Rthyass]{ThyAss}   {\Thy T \quad                \Term{\Delta}{\Phi}{\bool} \quad \Delta=\tp\valpha}        {\Thy{T,\ass\Pt\Delta \Phi}}\\[.5cm]
    \text{Contexts $\Delta$ and substitutions $\delta$ for them}\\
    \rul[Rctxempty]{CtxBase}{\Thy T}{\Ctx \circ} \quad
    \rul[Rctxtype]{CtxTp}{\Ctx \Gamma \quad \alpha\notin\Gamma}
    {\Ctx{\Gamma,\tp\alpha}} \quad
    \rul[Rsubempty]{SubBase}{\Ctx \Gamma}{\SubG \circ\bullet} \\[.3cm] 
    \rul[Rsubtype]{SubTp}{\SubG \delta\Delta \quad \alpha\notin\Delta \quad \TypeG{A}}
    {\SubG{\delta,A}         {\Delta,\tp\alpha}} \quad 
    \rul[Rctxterm]{CtxVar} {\Ctx \Gamma \quad x\notin\Gamma \quad \TypeG{A}}
    {\Ctx{\Gamma,x:A}} \\[.3cm] 
    \rul[Rsubterm]{SubVar} {\SubG \delta\Delta \quad x\notin\Delta \quad \TypeG{A}       \quad \TermG t{A[\delta]}}  {\SubG{\delta,t}         {\Delta,x:A}}\\[.3cm]
    
    \rul[Rctxass] {CtxAss} {\Ctx \Gamma \quad \TermG{\Phi}{\bool}}{\Ctx{\Gamma,\ass \Phi}} \quad
    \rul[Rsubass] {SubAss} {\SubG \delta\Delta \quad \TermG{\Phi}{\bool}\quad \TrueG {\Phi[\delta]}}   {\SubG{\delta,\checkmark}{\Delta,\ass \Phi}}\\[.5cm]
    \text{Lookup of type/term/assumption in a theory (left) or context (right)}\\
    \rul[Rtypesym]{TpSym}{\Ctx \Gamma \quad a:\Pg\Delta \type \in T\quad \SubG{\delta}{\Delta}}{\TypeG{a\ \delta }} \qquad
    \rul[Rtypevar]{TpVar}{\Ctx \Gamma \quad \alpha :\type\in\Gamma}{\TypeG \alpha}\\[.3cm]
    \rul[Rtermsym]{TermSym}{\Ctx \Gamma \quad c:\Pg\Delta A  \in T\quad \SubG{\delta}{\Delta}}{\TermG{c\ \delta}{A[\delta]}} \qquad
    \rul[Rtermvar]{TermVar}{\Ctx \Gamma \quad x:A\in\Gamma}{\TermG x A}\\[.3cm]
    \rul[Rvalsym]{ValidSym}{\Ctx \Gamma \quad \ass\Pg\Delta \Phi \in T\quad \SubG{\delta}{\Delta}}{\TrueG{\Phi[\delta]}} \qquad
    \rul[Rvalvar]{ValidVar}{\Ctx \Gamma \quad \ass \Phi\in\Gamma}{\TrueG \Phi}\\
  \end{array}\]
\caption{DHOL Rules for theories and contexts.}\label{fig:declrules}
\end{figure}

The rules in Fig.~\ref{fig:declrules} cover the structural parts, i.e., the declaration of and references to declarations in the theory and context.
Note how the rules for theory and context formation are parallel.
For example, \Rthytype and \Rctxtype handle the declaration of types $a:\Pt{\valpha}\P\vx\vA\type$ in a theory resp. $\tp\alpha$ in a context.
The main difference is that the former allows type symbols $a$, term symbols $c$, and global assumptions to depend on a list $\Delta$ of parameters.
To emphasize the common structure of the handling of parameters, we unify all parameter lists notationally into a single context $\Delta$.

Each of these six rules for making the declaration is paralleled by one of six rules for referencing (looking up).
For example, \Rtypesym and \Rtypevar reference type symbols $a$ resp. type variables $\alpha$.
Because the former is parametric in some context $\Delta$, their references must be instantiated with an appropriate substitution $\delta$ for the parameters.

Finally, the four rules for forming contexts are parallel to the four rules for forming substitutions.
For example, \Rsubtype shows how to extend a substitution $\delta$ for $\Delta$ to a substitution for $\Delta,\tp\alpha$ by adding a type substitute $A$ for $\alpha$.

\begin{figure}[htbp]
\[\begin{array}{l}
\text{Booleans: type formation, terms for equality and implication}\\
   \rul[Rtypebool]{TpBool}{\Ctx\Gamma}{\TypeG \bool} \quad
   \rul[Rtermeq]{TermEq}{\TermG tA \quad \TermG uA}{\TermG{\eq Atu}\bool}\\
   \rul[Rtermimpl]{TermImpl}{\TermG \Phi\bool \quad \Term{\Gamma,\ass \Phi}\Psi\bool}{\TermG{\Phi\impl\Psi}\bool}\\[0.5cm]
\text{Functions: type formation, $\lambda$-abstraction, application}\\
   \rul[Rtypepi]{TpPi}{\TypeG A \quad \Type{\Gamma,x:A}B }{\TypeG {\P xAB}}\\[.3cm]
   \rul[Rtermlam]{TermLambda}{\TypeG A \quad \Term{\Gamma,x:A}t B}{\TermG{\lam xAt}{\P xAB}}\quad
   \rul[Rtermapp]{TermApply}{\TermG{t}{\P xAB} \quad \TermG uA}{\TermG{t\ u}{B[u]}}\\[0.5cm]
\text{Type equality: congruence for all constructors and conversion of types}\\
   \rul[Rteqvar]{TpEqVar}{\Ctx \Gamma \quad \tp\alpha\in \Gamma}{\TEqG \alpha\alpha}\quad
   \rul[Rteqsym]{TpEqSym}{\Ctx \Gamma \quad a:\Pg\Delta \type \in T \quad \SubEqG{\Delta}{\delta}{\delta'}}{\TEqG{a\ \delta}{a\ \delta'}}\\[.3cm]
   \rul[Rteqbool]{TpEqBool}{\Ctx \Gamma}{\TEqG \bool\bool}\quad
   \rul[Rteqpi]{TpEqPi}{\TEqG A {A'}\quad \TEq{\Gamma,x:A}{B}{B'}}{\TEqG{\P xAB}{\P x{A'}{B'}}}\\[.3cm]
  \rul[Rtermteq]{TermConv}{\TermG tA \quad \TEqG A{A'}}{\TermG t{A'}}\quad
  \rul[Rboolext]{BoolExt}{\TrueG{t\ \bot} \quad \TrueG{t\ \top} \quad \top,\bot\text{ defined as usual}}{\True{\Gamma, x:o}{t\ x}} \\[.3cm]
     \text{where $\SubEqG \Delta\delta{\delta'}$ abbreviates the expression-wise provable equality}\\
  \text{of two substitutions for $\Delta$}\\[.5cm]  
\text{We omit the following routine validity rules: congruence rules for application;}\\
  \text{$\beta$, $\eta$ for functions; introduction and elimination for $\impl$}\\
  \text{Note that propositional and functional extensionality is implied}\\
  \text{by the rules for equality and the omitted eta rule~\cite{RRB23ext}.}
\end{array}
\]
\caption{DHOL Rules for Types and Terms}\label{fig:exprules}
\end{figure}

The rules in Fig.~\ref{fig:exprules} cover the rules for expressions.
We only point out some specialties: The rule \Rtermimpl makes implication $\Phi\impl\Psi$ dependent by assuming $\Phi$ while in the well-formedness of $\Psi$.
The rule \Rteqsym, while looking like a routine congruence rule, is the rule that makes type-checking undecidable by making two types equal if their arguments are equal component-wise.
Here the equality $\SubEqG \Delta\delta{\delta'}$ of substitutions holds if the term/type equality judgments hold for all corresponding pairs of terms/types in $\delta$ and $\delta'$.
And because term equality may depend on arbitrary assumptions in theory and contexts, so do all judgments.
Finally, \Rtermteq is only needed for constants and variables; for any other term, it can be derived using congruence.

\section{Translating PDHOL to PHOL}\label{sec:trans}
\paragraph*{Overview}

Like for DHOL the translation applies dependency erasure,
written with an overline $\tr{X}$, to turn PDHOL syntax into PHOL syntax.
Note that because the latter is a fragment of the former, we can reuse all PDHOL notations to write PHOL
syntax.
Intuitively, term-dependent types are translated into types without their term arguments.
The lost information is captured in a partial equivalence relation (PER) in the style of \cite{persemantics}.
Readers may consult the examples below or the invariants stated in Thm.~\ref{thm:preserve} to help their intuitions.

All term arguments of type symbols are erased; type arguments are kept, i.e., polymorphic symbols remain polymorphic.
In particular, we translate the type $a\ \vA\ \vt$ to $a\ \tr{\vA}$ and the type $\P xAB$ to $\tr A\to \tr B$.
To recover the erased typing information, we define for every PDHOL-type $A$ a
PER $\per A$ on $\tr{A}$ in PHOL
such that the PHOL-formula $\perapp A {\tr t}{\tr t}$ captures the PDHOL-typing judgment $t:A$.
Critically, term equality \(s =_A t\) is translated to $\perapp A {\tr s}{\tr t}$.

PERs are symmetric and transitive relations, and an element in a PER is related to itself iff it is related to any element.
So $\per A$ is an equivalence relation on a subtype of $\tr A$.
The corresponding quotient of that subtype of $\tr A$ is the semantics of $A$ under our translation.
We write $\isper r$ to abbreviate that $r$ is a PER (i.e., symmetric and transitive).

Expanding the usual definitions of the quantifiers reveals that (up to provable equivalence)
$\perapp Axx$ acts as a guard when quantifying over $x$. 
One might think that a unary predicate (indicating a subtype) would suffice as a type guard instead of a binary predicate.
But using quotients of subtypes (and thus PERs) becomes necessary at higher function types where two functions are equal if they map type-guarded inputs to equal outputs.

\paragraph*{Auxiliary Notations}
While conceptually straightforward, binding the parameters in theories in all generated
declarations is notationally complex. Therefore, we abbreviate as follows:

\begin{definition}[Abbreviations for PHOL-Contexts]\label{def:holabbrs}
Given a PHOL-context $\Gamma$ and a PHOL-substitution \(\gamma\), we abbreviate
\begin{itemize}
\item $\ot\Gamma$ is like $\Gamma$ but contains only the t\textbf{y}pe variables.
\item $\ott\Gamma$ is like $\Gamma$ but contains only the type and t\textbf{e}rm variables. (In HOL, as opposed to DHOL, such taking of subcontexts is always legal as long as the order is preserved.)
\item $\ot\gamma$ is like $\gamma$ but contains only the t\textbf{y}pe substitutes.
\item $\ott\gamma$is like $\gamma$ but contains only the type and t\textbf{e}rm substitutes.
\item If $\Gamma$ consists of only type variables, we write $\Gamma\to\type$ for the kind $\type\to\ldots\to\type\to\type$ (taking one type argument for every type variable in $\Gamma$).
\item If $\Gamma$ consists of type variables $\valpha$ and term variables $\vx:\vA$, we write $\Gamma\to B$ for $\Pt\valpha A_1\to\ldots\to A_n\to B$.
\item If $\Gamma$ consists of type variables $\valpha$, term variables $\vx:\vA$, and assumptions
  $\ass \Phi_i$, we write $\Gamma\impl \Psi$ for the PHOL-formula $\Pg\valpha \forall
  x_1:A_1.\ldots\forall x_m:A_m.\Phi_1\impl \ldots \Phi_n\impl \Psi$; if $\Gamma$ contains alternating term variables and assumptions, we alternate the $\forall$ and $\impl$ bindings accordingly.
\end{itemize}
\end{definition}

The abbreviations $\ot\Gamma$ and $\ott\Gamma$ remove parameters that a PHOL-declaration cannot bind: term symbol declarations cannot bind assumptions, and type symbol declarations cannot bind assumptions or term variables.
These occur because every type $A$ yields a translated type $\tr A$, a binary relation $\per A$, and a statement $\isper{\per A}$.
Consequently, a type variable $\alpha$ is translated to three declarations: a type variable $\alpha$, a binary relation $\alpha^*$ on it, and a local assumption $\isper{\alpha^*}$.
Similarly, every term $t:A$ yields a translated term $\tr t$ and a statement $\perapp A {\tr t}{\tr t}$.
Consequently, every term variable $x$ is translated to two declarations: a term variable and a local assumption.
Thus, $\tr\Gamma$ contains a mixture of type variables, term variables, and assumptions even if $\Gamma$ contains only type variables.
The latter must be removed if we want to bind $\tr\Gamma$ in a PHOL-type or term symbol declaration.

The abbreviations $\Gamma\to\type$, $\Gamma\to A$, and $\Gamma\impl \Phi$ efficiently perform these bindings.
For example, if a PDHOL axiom $\ass \Pt\valpha \Phi$ is parametric in type variables $\alpha_1,\ldots,\alpha_n$, then $\tr\Gamma$ contains $3n$ declarations.
$\tr\Gamma \impl \tr \Phi$ binds all of them to construct a PHOL-axiom: the type variables in $\tr\Gamma$ remain type variables, the term variables are $\forall$-bound, and the unnamed assumptions are bound by $\impl$.
Any PDHOL usage of this axiom instantiates each type variable $\alpha$ with a type $A$.
In PHOL, this corresponds to instantiating $\alpha$ with $\tr A$, universal elimination with $\per A$, and implication elimination with $\isper{\per A}$.

\paragraph*{Formal Definition}
We translate types and terms as follows:

\begin{center}
\begin{tabular}{|l|l||l|l|}
\hline
PDHOL type $A$ & Translation $\tr A$ in PHOL & PDHOL term $t$ & Translation $\tr t$ in PHOL\\
\hline
$a\ \delta$ & $a\ \ot{\tr{\delta}}$ & \(c \ \delta\) & \(c\ \ott{\tr{\delta}}\) \\
$\alpha$ & $\alpha$ & $x$ & $x$ \\
$\bool$ & $\bool$ & $\eq A t u$ & $\perapp Atu$ \\
                 && $\Phi\impl\Psi$ & $\tr \Phi \impl \tr \Psi$\\
$\P xAB$ & $\tr A \to \tr B$ & $\lam xAt$ & $\lam x{\tr A}{\tr t}$ \\
                            && $t\ u$ & $\tr t \ \tr u$ \\
\hline
\end{tabular}
\end{center}
For the translation of type and term symbol \emph{references}, compare the matching translation of the corresponding \emph{declarations} in theories below.
Contexts (left) and substitutions (right) are translated by concatenating the translations of their components:
\begin{center}
\begin{tabular}{|l|l||l|l|}
\hline
PDHOL & Translation & PDHOL & Translation\\
\hline
$\tp\alpha$ & $\tp\alpha$, $\alpha^*:\alpha\to\alpha\to\bool$, $\ass \isper{\alpha^*}$  & $A$   & $\tr A$, $\per A$, $\checkmark$\\
$x:A$       & $x:\tr A$, $\ass A^*\ x\ x$ & $t$   & $\tr t$, $\checkmark$\\
$\ass \Phi$    & $\ass \tr \Phi$  & $\checkmark$ & $\checkmark$ \\
\hline
\end{tabular}
\end{center}
Note how substitutions for $\Delta$ are translated to substitutions for $\tr \Delta$: for example, just like every type variable produces $3$ declarations, every type substitute produces $3$ corresponding substitutes.
In particular, each $\checkmark$ in the translation of a substitution represents the invariant that the respective formula is in fact provable in the translated context:
for example, for every PDHOL-type $A$, PHOL can prove $\isper{\per A}$, and for every PDHOL-term $t:A$, PHOL can prove $\perapp Att$.
These properties are shown in Thm.~\ref{thm:preserve}.

The definition of the PER follows the principles of logical relations:
\begin{center}
\begin{tabular}{|l|l|}
\hline
PDHOL type $A$ & $\perapp Atu$ in PHOL\\
\hline
$a\ \delta$ & $a^*\ \ott{\tr{\delta}}\ t\ u$ \\
$\alpha$ & $\alpha^*\ t\ \ u$ \\
$\bool$ &  $\eq\bool tu$ \\
$\P xAB$ & $\forall x,y:\tr A. \perapp Axy\impl \perapp B{(t\ x)}{(u\ y)}$ \\
\hline
\end{tabular}
\end{center}
Here $a^*$ and $\alpha^*$ are names introduced during the translation of theories and contexts.
These declare the respective PER axiomatically for type symbols and type variables.

\begin{example}
  \label{ex:trans1}
  We use the expression \(E\) from Ex.~\ref{ex:term} to create \(E_2\), a list \texttt{[E, E]} of lists by \(E_2 := \cons\ (\vect\ \nat\ 2)\ 1\
  E\ (\cons\ (\vect\ \nat\ 2)\ 0\ E\ (\nil\ (\vect\ \nat\ 2))):\vect\ (\vect\ \nat\ 2)\ 2\).
  Its translation \(\tr{E_2}\) yields \(\cons\ (\vect\
  \nat)\ 1\ \tr{E}\ (\cons\ (\vect\ \nat)\ 0\newline \tr{E}\ (\nil\ (\vect\ \nat)))\).
  
  Here \(\delta\) is \((\vect\ \nat\ 2)\ 1\ E\ (...)\). Observe that
  the erasure recurses through the argument list, i.e., \(\tr{(\vect\ \nat\ 2)}\ \tr{1}\ \tr{E}\ 
  \tr{(...)}\). The type argument \(\vect\ \nat\ 2\) is translated into \(\vect\ \nat\)
  removing the term argument to the type as well as the generated PER \((\vect\ \nat\ 2)^*\) in
  accordance with our definition of \(\ot\delta\). The remaining arguments are terms that do not
  take term arguments, meaning that the translation does not change them.

  The type of the expression \(\vect\ (\vect\ \nat\ 2)\ 2\) is correspondingly erased to
  \(\vect\ (\vect\ \nat)\).
\end{example}

Finally, the cases for declarations in theories are more complex because they contain contexts.
This is where the abbreviations from Def.~\ref{def:holabbrs} are used:

\begin{center}
\begin{tabular}{|ll|l|}
\hline
PDHOL && Translation in PHOL\\
\hline
$a:\Pg\Delta\type$ &where $\Delta=\tp\valpha,\ \vx:\vA$ & $a:\ot{\tr\Delta}\to\type$ \\
     && $a^*:\ott{\tr\Delta}\to a\ \valpha \to a\ \valpha\to\bool$ \\
     && $\ass \tr\Delta \impl\isper{a^*\ \valpha\ \valphaper\ \vx}$ \\
\hline
$c:\Pg\Delta A$ &where $\Delta=\tp\valpha$            & $c:\ott{\tr{\Delta}}\to\tr A$ \\&& $\ass \tr{\Delta}\impl \perapp A {(c\ \valpha)} {(c\ \valpha)}$\\
\hline
$\ass \Pg\Delta  \Phi$ &where $\Delta=\tp\valpha$        & $\ass \tr{\Delta} \impl \tr \Phi$ \\
\hline
\end{tabular}
\end{center}

\begin{example}
  \label{ex:trans2}
  Translating our running example's theory yields \(\vect:\type\to\type\) for the vector
  declaration. Note that while \(\Delta\) includes type and term variables, the first part of
  the erasure only considers \(\ot\Delta\), doing away with the term variables, therefore
  \(\ot{\tr\Delta}\to\type\) is the kind that takes an equal number of type variables as arguments
  and returns a type.

  The second part of the erasure of vector yields
  \(\per{\vect}:\Pt{\alpha}(\alpha\to\alpha\to\bool)\to\nat\to(\vect\
  \alpha)\to(\vect\ \alpha)\to\bool\). \(\ott\Delta\) includes, additionally to the type variables
  from the previous point, the term variables. This includes the PER generated by the erasure
  of the type variables as well as the term argument \(\nat\).

  Finally, we get the axiom \(\ass
  \Pt{\alpha}\forall\per{\alpha}:\alpha\to\alpha\to\bool.\forall n:\nat.PER(\per{\alpha})\impl PER(\per{\vect}\
  \alpha\ \per{\alpha}\ n)\). 
  Compared to the last two results of the erasure, we now have additionally the assertion that
  the type argument comes equipped with a PER. As this is now a validity statement (as opposed
  to a typing statement) term arguments are now bound by \(\forall\) and \(\impl\).
  
  Readers may note the absence of the according statements for the type
  \(\nat\). In the case of non-dependent base types, the PER collapses to standard equality, resulting in trivial axioms, which we omit.
\end{example}

Our translation subsumes that of \cite{RRB23ext}: If we specialize to DHOL,
contexts must not contain type variables and the corresponding cases in the translation of contexts and substitutions and the definition of the PER can be dropped.
The arguments $\Delta$ of a type symbol declaration $a$ in a theory contain only type
variables, and references $a\,\delta$ take only type arguments, so that (i) $\ot{\tr\Delta}$
and $\ot{\tr\delta}$ are empty (ii) $\ott{\tr\Delta}$ contains only the declarations $x:\tr A$
for every type $A$ in $\Delta$ (iii) $\ott{\tr\delta}$ contains only the terms $\tr t$ for every term $t$ in $\delta$

Finally, the argument list $\delta$ of a term symbol $c$ is empty and thus so is $\ott{\tr\delta}$.

\paragraph*{Properties of the Translation}\label{sec:meta}
\begin{theorem}[Preservation of Judgments and Substitution]\label{thm:preserve}
Every PDHOL judgment in the table below implies the corresponding PHOL judgment about the translated syntax:
\begin{center}
	\begin{tabular}{|l|l|}
		\hline 
		PDHOL & PHOL \\ 
		\hline
		$\Thy T$ & $\Thy {\tr T}$\\
		$\Ctx \Gamma$ & $\Ctx[\tr T]{\tr\Gamma}$\\
		$\SubG \delta \Delta$ & $\Sub[\tr T]{\tr\Gamma}{\tr\delta}{\tr\Delta}$\\
		$\TypeG A$ & $\Type[\tr T]{\tr\Gamma}{\tr A}$ and \\
		   & \quad $\Term[\tr T]{\tr\Gamma}{\per A}{\tr A\to\tr A\to\bool}$ and $\True[\tr T]{\tr\Gamma}{\isper{\per A}}$\\ 
		$\TermG tA$ & $\Term[\tr T]{\tr\Gamma}{\tr t}{\tr A}$ and $\True[\tr T]{\tr\Gamma}{\perapp{A}{\tr t}{\tr t}}$\\
		$\TrueG F$ & $\True[\tr T]{\tr\Gamma}{\tr F}$ \\
		$\TEqG AB$ & $\TEq[\tr T]{\tr\Gamma}{\tr A}{\tr B}$ and  $\True[\tr T]{\tr\Gamma}{\eq{\tr A\to\tr A\to\bool}{\per A}{\per B}}$\\ 
		\hline
	\end{tabular}
 \end{center}
 
Moreover, whenever $\SubG \delta \Delta$, we have
\begin{center}
	\begin{tabular}{|l|l|}
		\hline 
		PDHOL & PHOL \\ 
		\hline
		$\Type{\Gamma,\Delta}A$ & $\TEq[\tr T]{\tr\Gamma}{\tr{A[\delta]}}{{\tr A}[{\tr\delta}]}$ \\
		$\Term{\Gamma,\Delta}tA$ & $\True[\tr T]{\tr\Gamma}{\eq{\tr{A[\delta]}}{\tr{t[\delta]}}{{\tr t}[\tr\delta]}}$\\ 
		\hline
	\end{tabular}
\end{center}
\end{theorem}
The proof of Theorem~\ref{thm:preserve} is straightforward and performed by induction over
derivations. An illustrative example is given in Appendix~\ref{app:comp}.

\begin{theorem}[Reflection of Truth]\label{thm:sound}
Assume a well-formed PDHOL theory \(\Thy{T}\). 
\begin{center}
    If \(\Gamma \vdash_T F:o\) and \(\tr{\Gamma} \vdash_{\tr{T}} \tr{F}\) then \(\Gamma\vdash_T F\).
\end{center}
In particular, if \(\Gamma \vdash_T s:A\) and \(\Gamma \vdash_T t:A\) and \(\tr{\Gamma} \vdash_{\tr{T}} \perapp{A}{\tr{s}}{\tr{t}}\), then \(\Gamma \vdash s =_A t\).
\end{theorem}
The assumption about the well-typedness of the statement is necessary: Consider two
PDHOL-terms \(u := \lambda x:A\ s.x\) and \(v := \lambda x:A\ t.x\) of some dependent type
\(a\) and different terms \(s, t\). In (P)DHOL an equality \(u =_{\Pi x:A\ s.A\ s} v\) between
them would be ill-typed. The erased equality however is not only well-typed but provable.

    The original proof for DHOL~\cite{RRB23ext} is rather involved.
    Luckily it is easy to extend it to PDHOL.
    Intuitively, a proof of a PHOL statement \(\tr{F}\) is translated into a proof that exists in the image of the translation.
    This allows us to subsequently read off a PDHOL proof of the untranslated conjecture
    \(F\). An overview of the original proof together with the necessary adaptations is given
    in Appendix~\ref{app:sound}.
In the remainder, we will call the combination of these properties \emph{well-behavedness} of
the translation.

\section{Implementation and Case Studies}\label{sec:case}





To evaluate PDHOL and obtain a theorem prover for it, we implemented our translation as part of the
logic-embedding tool by Steen~\cite{leo_embedding}.
This enables discharging PDHOL proof obligations by existing PHOL ATPs.
Our tool, as well as the dialect used to formalize this set of problems, has since been adopted as TPTP standard in the form of the new DHF dialect and the DT2H2X prover on
SystemOnTPTP~\cite{R+25}.

As outlined in Section~\ref{sec:meta}, the translation requires the problem to be well-typed.
Therefore, the use of a PHOL ATPs is only sound if we first obtain a type-checker for PDHOL.
Because type-checking for PDHOL is undecidable, our tool transforms each conjecture $F$ into $n+1$ PHOL conjectures:
$n$ type-checking obligations (TCOs) that establish the well-typedness of the problem plus $F$ itself.

As an optimization, we replace PERs in our translation with equality whenever there is no dependence on the term arguments.
From an informal comparison with previous data, this seems to yield a speedup of about 5\% on
the presented problems.

We created a set of 53 problems to evaluate and test our implementation on.
These include 18 problems that are polymorphic versions of the examples created in~\cite{dhol_satallax} for DHOL, and 17 problems that experiment with the impact of instantiating type variables in those conjectures.
11 more problems are lemmas that occur during an inductive proof that the reversal function on
red-black trees is an involution (see below).
Additionally, there are one formalization each of the zip function as used in functional
programming and of finite sets.
The remaining problems focus on the formalization capabilities of PDHOL as presented in
Section~\ref{sec:std}, in the form of the group problem. This is not yet supported by the
implementation.

Three PHOL ATP systems were combined with our translator to create PDHOL provers: Vampire
v4.7~\cite{vampire} in portfolio mode, Leo-III v1.6~\cite{leo3_modern} and Zipperposition
v2.1~\cite{C15} in higher-order and portfolio mode.
The reported times do not include problem translation time, which amounted to
198\(\pm\)33ms.
All experiments ran on an Intel Core i5-6200U CPU at 2.30GHz and 8 GB of RAM.
All files and binaries, including the group example from Section~\ref{sec:std}, can be
found in the supplementary material.\footnote{\url{https://zenodo.org/records/19931555}}

In the following, we give an overview of those problems. The names of the TPTP files as can be
found in the supplements are given in parentheses, with asterisks acting as wildcards.

\makecn{fin}
\makecn{nat}
\makecn{elemAt}
\makecn{transpose}
\paragraph*{Fixed-Length Lists}
We did a deep formalization of lists following the examples in Section~\ref{sec:syn} and \ref{sec:trans}.
We also formalized finite types $\fin:\nat\to\type$, and used an accessor $\elemAt:\Pg\alpha\P n\nat\vect"{\alpha,n}\to\fin"n\to\alpha$ for safe access to list elements.
The conjectures we investigated expressed properties of the 
append function --- associativity (\texttt{list-app-assoc*}) and the identity of nil
(\texttt{list-app-nil*}) --- as well as the  
involution property of the reverse function (\texttt{list-rev-nil*}). 
These problems generated more complex TCOs than the red-black trees, due to the arithmetic
needed to show that 
lengths of lists line up after appending.
From this, matrices can then be represented by nesting vectors.
We formalize, e.g., matrix transposition as $\transpose:\Pg\alpha\P {m,n}\nat
\vect"{(\vect"\alpha"m)}"n\to\vect"{(\vect"\alpha"n)}"m$, and formalize commutativity of matrix
addition (\texttt{inst-poly-matrix-add-com}) and element-wise multiplication
(\texttt{inst-poly-matrix-mul-com}).

This leverages both dependent types and polymorphism to concisely represent complex data
structures in a way that guarantees their dimensional invariants.

Another formalization leveraging the list basis, is our formalization of the zip function.
Note that Haskell's default List
package\footnote{\href{https://hackage.haskell.org/package/base-4.21.0.0/docs/Data-List.html\#v:zip3}{hackage.haskell.org/package/base-4.21.0.0/docs/Data-List.html}}
includes distinct functions \texttt{zipN} for \(N \in \{3,\ldots,7\}\) arguments due to the lack
of dependent types. Our formalization (\texttt{zip}) acts like an arbitrary dimensional zip
function.

Adding a formalization of the finite set (\texttt{finite-type}) allows to have a unfailing
getter function. This example is set up in a way that type-checks, but where the types do not
make sense: They would allow getting an element of the empty list. Indeed, as can be seen in
the results in Appendix~\ref{app:res}, Leo-III returns the SZS Status ``Contradictory Axioms''.

\vspace{-3mm}
\paragraph*{Red-Black Trees}
Red-black trees are binary search trees for which self-balancing is 
achieved by coloring 
nodes
red or black.
Leaves are always black. In addition, two
constraints are maintained: \emph{The child of a red node must not be red} and
\emph{Every path from a node to its leaves has the same number of black nodes}.
Such invariants are difficult to capture by non-dependent inductive types.
However, in PDHOL, we can use a declaration 
\(\text{rbtree : } \Pt{\alpha}\P{b}{\bool}\P{n}{\nat} \type\) 
where
\(\text{rbtree}\ A\ b\ n\)
is the type of red-black trees holding values of type \(A\)
containing \(n\) black nodes. This allows capturing the invariant directly in the type.
The formalized conjecture states that reversing is an involution.


The standard proof of this fact involves inductive definitions that are difficult
for ATP systems to deal with. In line with previous work in DHOL~\cite{dhol_satallax,RBK24} we
split the problem into smaller sub-problems, ensuring well-typedness.

The problem is split into a base case for red (\texttt{rbt-rev-invol-base-RTLeaf}, the prefix
\texttt{rbt-rev-invol} will be omitted from here on) and black (\texttt{-base-BTLeaf})
leaves. Additionally, there is a problem file asserting 
that, if a property holds for red leaves and black leaves, it holds for all red-black trees of
black-height 1 (\texttt{-base-lemma}). This lemma is used to prove the base case
for any tree of black-height 1 (\texttt{-base}).

Next are the step cases: Again there is a step case for red (\texttt{-RTStep}) and for black
(\texttt{-BTStep}) leaves. While the red
step case is easy, the black is challenging due to the raised complexity of black nodes. While
possible to prove problem \texttt{rbt-rev-invol-BTStep} directly, we found three sub-steps
(\texttt{-subBTStep[1-3]}) and added them as axioms to help the ATP system along
(\texttt{-composedBTStep}). 

While the instanced induction scheme (\texttt{-indinst}) timed out, type checking
was successful. Assuming the instanced induction allows to proof the main conjecture
(\texttt{rbt-rev-invol}) in just 14s.


Finally, we specified one problem not related to the reversal function (\texttt{bad\_tree}). It serves as a sanity
check for the formulation of red-black trees: a conjecture trying to create an ill-formed
red-black tree. It postulates that for every red tree, there are two children of arbitrary
color. This violates the invariant that a red node must not have red children. Indeed, neither
the conjecture itself nor the generated TCOs 
can be proven. 

Notably, the red-black tree examples generate very few TCOs (4
in total) compared to the list examples, where the majority of examples generated 1+ TCOs.
This is 
because most constraints, thanks to the efficient formulation 
provided by the dependent types, reduce to reflexivity. 

\paragraph{Results and Discussion}
The results of our experiments can be found in Appendix~\ref{app:res}.
A significant number of statements can be proven. This is notably without any fine-tuning of
the provers to the specific challenges translated PDHOL problems pose. This illustrates the
viability of the translate-and-proof method and, in particular, shows that PDHOL retains
automation support while providing a rich addition to expressivity.

Notably, the results support that undecidable type-checking is a price worth paying --- no TCOs
were left unproven. Frequently the TCOs were trivial and required no call to the ATP
systems. While this will not be the case for every problem, it is promising to see that for
these common data-structures type-checking is efficient.

However, for many advanced conjectures, e.g., showing that transposing matrices is an
involution, the proofs timed out. 
We believe this is because the necessary induction arguments are too difficult for current ATPs (independent of how the data structures are formalized). 
This is in line with the results reported by Niederhauser et al.~\cite{dhol_satallax},
indicating that performance improvements can be obtained by building DHOL-specific and/or
induction-aware provers, or
splitting problems in a way that helps the ATP system with the induction arguments.

\section{Dependent Type Variables}\label{sec:dtv}

%
%
%

For simplicity, we have so far not described the feature of type variables depending on term variables as in $\tp{\alpha},\,\beta:\alpha\to\type$.
The following table shows the possible dependencies for variables in \emph{contexts} akin to how we did it in Section~\ref{sec:syn} for \emph{theories}:
\begin{center}
\begin{tabular}{|l|l|l|l|}
\hline
 & \multicolumn{3}{c|}{declaration of} \\
depending on & type variable & term variable & local assumption \\
\hline
type variable & \multicolumn{3}{c|}{not allowed to avoid size issues} \\
term variable & this section & definable via $\Pi$ & definable via $\forall$ \\
local assumption&  \multicolumn{2}{c|}{not allowed} & definable via $\impl$\\
\hline
\end{tabular}
\end{center}
Contrary to global declarations in theories, variable declarations must not bind type variables.
This ensures that contexts are small and we can formulate the semantics without requiring a hierarchy of universes or similar.
After adding the final feature of dependent type variables, the grammar of PDHOL becomes:
\begin{center}
\begin{tabular}{|l c l | l c l|}
  \hline
  \(T\) & ::= & \multicolumn{4}{l|}{\(.\) | $T, a:\Pg{\hldtv{\Gamma}}\type$ | $T, c: \Pg{\hldtv{\Gamma}} A$ | $T, \ass\Pg{\hldtv{\Gamma}} F$}\\
  $K$   & ::= & $\type$ | $\hldtv{\P xA} K$ & $k$ & ::= & $A$ | $\hldtv{\lam xA} k$ \\

  \(\Gamma\) & ::= & \(\circ\) | $\Gamma, \alpha:K$ | $\Gamma, x: A$ | $\Gamma,\ass F$ & \(\gamma\) & ::= & \(\bullet\) | $\gamma, k$ | $\gamma, t$ | $\gamma,\checkmark$ \\
  \(A,B\) & ::= & $a\:\gamma$ | $\alpha\:\hldtv{t\ \ldots\ t}$ | $\P xAB$ | $\bool$ & \(t,u\) &
                                                                                                ::= & $c\:\gamma$ | \(x\) | $\lam xAt$ | $t\:u$ | $t\Rightarrow u$ | $\eq Atu$ \\
  \hline
\end{tabular}
\end{center}
Here $k$ produces types with uninstantiated term dependencies, and $K$ produces their kinds.
The grammar used so far arises as the special case $K=\type$ and $k=A$.
Their typing is given by the rules
\[\rul[Rtermlam']{}{\Term{\Gamma,\Delta}B\type}{\TermG{\lamg\Delta B}{\Pg\Delta\type}} \quad
\rul[Rtypesym']{}{\alpha:\Pg\Delta\type\in \Gamma \quad \SubG \delta\Delta}{\TypeG{\alpha\ \delta}} \quad \text{where } \Delta=\vx:\vA\]
Adjusting \Rctxtype and \Rsubtype accordingly is straightforward.

The grammar above describes the arguments of global declarations as a single context (as in $\Pg\Gamma$) rather than a list of type variables followed by a list of term variables (as in $\Pt \valpha\P\vx\vA$).
This is necessary because term variables $x:A$ may now occur in a type variable declaration $\alpha: A\to \type$.
Therefore, we can no longer assume that all type variables are bound before all term variables.
The typing rules for theories remain unchanged except that we to modify the syntactic restrictions on the context in \Rthytype, \Rthyterm, and \Rthyass.
For the same reason as for global declarations in theories, we do not allow local assumptions in $\Gamma$.

\begin{example}
We give heterogeneous lists as a variant of vectors where each component may have a different type.
Consequently, all operations must take a dependent type variable $\alpha:\fin\ n\to\type$ that provides the types of the $n$ components.
This uses finite types $\fin\ n=\{0,\ldots,n-1\}$, which we can declare by
\begin{gather*}
  \cn{fin} :\nat \to \type \qquad
  \fnext: \P n\nat \fin\ n \to \fin{(\suc\ n)}\\
  \ftop: \P n\nat \fin\ {(\suc\ n)}
\end{gather*}
Then heterogeneous lists can be declared as
\begin{gather*}
  \cn{Hlist}: \P n\nat (\fin\ n \to \type) \to \type \\
  \hlist: \P n\nat \P L{\fin\; n \to\type} (\fin\ n \to L\ n) \to \Hlist n L \\
  \hget:  \P n\nat \P L{\fin\; n \to\type} \Hlist n L \to \P i {\fin\; n} L\ i
\end{gather*}
Note how type and term variables alternate, e.g., in the declaration of $\hlist$.
\end{example}

The translation rules for declarations in a theory remain unchanged except for generalizing the syntactic restrictions on the contexts and substitutions.
We only need to adjust the three cases below for the translation of type variables.
The invariants are the same.
\begin{center}
\begin{tabular}{|ll|l|}
\hline
& PDHOL  & PHOL\\
\hline
type variable declaration & $\alpha:\Pg\Delta\type$ where $\Delta=\vx:\vA$ & $\tp\alpha$ \\&& $\alpha^*:\ott{\tr\Delta}\to\alpha\to\alpha\to\bool$, \\&& $\ass \tr\Delta\impl\isper{\alpha^*\ \vx}$ \\
type variable substitute & $\lamg\Delta A$ where $\Delta=\vx:\vA$ & $\tr A,\; \lamg{\ott{\tr\Delta}} \per A$,\; \checkmark \\
type variable reference & $\alpha\ \delta$ & $\alpha$\\
\hline
\end{tabular}
\end{center}

\section{Subtype Definitions}\label{sec:std}

Large HOL theories are usually built by chaining conservative extension principles.
Besides direct definitions, the most important such principle used in HOL-based ITPs is definitional subtypes \cite{holsemantics}.
For example, the theory-level declaration $a:=A|p$ for some predicate $p:A\to\bool$ conservatively extends the containing theory with declarations for a fresh (undefined) type $a$ and fresh functions $Abs:A\to a$ and $Rep:a\to A$ that are axiomatized to be bijections between $a$ and the set of elements of $A$ that satisfy $p$.
This extension principle is expressive enough to define, e.g., record and inductive types.
But it becomes particularly powerful in the presence of polymorphism, where it can introduce new type \emph{operators} $a$.
For example, simple product types can be defined using $a\,\alpha\,\beta:=(\alpha\to\beta\to\bool)|p$ where $p$ is the property that $f$ is true for exactly one argument pair.

Generalizing this extension principle to dependent types further increases its expressivity and enables PDHOL to make complex type definitions.
We extend our language as follows:
\[T\;\;  ::=\;\; T,\, a:=\lamg{\Delta}A|p \quad\text{for some}\quad \Delta=\tp\valpha,\,
  \vx:\vA\]
with a typing rule requiring $\Term\Delta p{A\to \bool}$ and a proof obligation that we discuss below.
Abbreviating $\delta:=\valpha\,\vx$, its semantics is defined by elaboration into
\newcommand{\aq}{a\,\delta}
\newcommand{\Repq}[1]{Rep"{\delta,#1}}
\newcommand{\Absq}[1]{Abs"{\delta,#1}}
\begin{align*}
a:\Pg\Delta \type& \qquad
  Abs: \Pg\Delta A\to \aq \qquad
  Rep: \Pg\Delta\aq\to A\\
\ass \Pg\valpha \forall\vx:\vA.&\big(\forall u:\aq.\;p"{(\Repq u)}\wedge\eq{\aq}{\Absq{(\Repq
  u)}}u\big)\\
  \wedge \;&\big(\forall v:A.\;p"v\Rightarrow \eq A{\Repq{(\Absq v)}}v\big).
\end{align*}
In the same way as for HOL subtype definitions, in the second part of the axiom, the predicate $p$ occurs crucially to require $\eq A{\Repq{(\Absq v)}}v$ only for those $v$ that are in the intended subtype, capturing that we think of $Abs$ as a partial function.
Because all functions are total, $\Absq v$ is also well-typed if $\neg (p\,v)$ but remains unspecified.
That is conservative if the new type $a$ is non-empty, which is why HOL additionally requires the proof obligation $\exists v:A.p\,v$.
While this is natural for HOL (where all types are assumed to be non-empty anyway), this is not ideal for DHOL, where empty types are useful and allowed.
Nonetheless, we adopt the same proof obligation here for simplicity, i.e., a subtype definition is well-typed only if $\True\Delta{\exists v:A.p\,v}$.
Our translation is in fact judgment preserving and truth reflecting for weaker conditions, but we leave the details to future work.

We extend our translation as follows: every subtype definition $a:=\lamg{\Delta}A|p$ is translated to the corresponding PHOL subtype definition $a:=\lamg{\valpha} \tr A|(\lam u{\tr A}\perapp Auu\wedge p\,u)$.
This translation commutes with elaboration: we obtain isomorphic PHOL theories if we (i)
elaborate a PDHOL subtype definition and then translate it 
, or (ii) translate it to a PHOL definition and then apply the usual PHOL elaboration.
Consequently, 
the translation remains well-behaved.

\makecn{group}
\makecn{isgroup}
\makecn{ishom}
\makecn{hom}
\makecn{neut}
\makecn{op}
\newcommand{\Rep}[1]{\mathrm{Rep}_{#1}}
\newcommand{\Abs}[1]{\mathrm{Abs}_{#1}}

\begin{example}[Algebraic Structures]
Let us assume we have already formalized types for algebraic structures, such as a type $\group:\Pg\alpha\type$ with (among others) a selector $\op:\Pg\alpha\alpha\to\alpha\to\alpha$.
(In principle, the type $\group"\alpha$ could itself be defined as a subtype of $\alpha\to\alpha\to\alpha$.
But that would require a more complex analysis of conservativity because there is no group on the empty type.)

We can now use polymorphism to abstract over the carrier set and then use dependent types to formalize various constructions from universal algebra.
For example, we can define the predicate $\ishom:\Pg{\alpha,\beta}\,\group\,\alpha\to
\group\,\beta\to (\alpha\to\beta)\to\bool$ 
, and use that to define the polymorphic and dependent type of group homomorphisms \(\hom\) by
\begin{align*}
  \eq\bool{\ishom"{\alpha,\beta,G,H,m}}{\forall
    x,y:\alpha. \eq\beta{m"{(\op"{G,x,y})}}{\op"{H,(m"x),(m"y)}}}\\
  \hom:=\lambda_{\tp\alpha,\,\tp\beta,\,G:\group"\alpha,\,H:\group"\beta}
  (\alpha\to\beta)|\ishom"{\alpha,\beta,G,H}
\end{align*}
Similar examples are the types of subgroups of a group, conjugacy classes of a group, or actions of a group on a set, and accordingly for all other algebraic theories.

As an example theorem, we state that homomorphisms preserve the neutral element:
\begin{align*}
\Pg{\alpha,\beta}\;\forall
  G:\group"\alpha,\,H:\group"\beta,\,m:\hom"{\alpha,\beta,G,H}.\;\forall
  e:\alpha.\neut"{\alpha,G,e}\\
  \Rightarrow\neut"{\beta,H,(\Rep\hom"{m,e})}
\end{align*}
This example illustrates the expressiveness of the PDHOL language while assuaging lingering
doubts about undecidable type checking: Although its a complex formulation we cannot proof directly,
similar to examples in Section~\ref{sec:case}, the TCOs are discharged easily. 
\end{example}



\section{Conclusion and Related Work}
\paragraph*{Summary}
We extended dependently typed higher-order logic with polymorphism and subtype definitions.
A translation to polymorphic HOL yields an efficient automated reasoning procedure for the calculus.
We demonstrated the practical usability of the language and its automation by encoding and automatically proving several practical problems that combine polymorphism with dependent types and cannot be represented as concisely in either polymorphic HOL or monomorphic DHOL.

\paragraph*{Related Work}
While there are some practical examples that make deep polymorphism desirable, they tend to interact poorly with ATP systems.
Indeed, most logics that support deep polymorphism introduce a hierarchy of universes as in Martin-L\"of type theory \cite{martinlof}, which goes far beyond the expressivity of current ATP tools.
On the other hand, shallow polymorphism still allows standard set-theoretic semantics without any size issues.
That makes it the variant of choice in HOL theorem provers --- both interactive~\cite{hol,isabellehol,hollight} and automated~\cite{SWB17,vampire,V+22}.
Indeed, our definition has the key advantage that DHOL polymorphism can be directly translated to HOL polymorphism, ensuring that proof obligations remain efficiently solvable for the ATP system.

PVS \cite{pvs} provides a combination of dependent types and shallow polymorphism 
similar to PDHOL.
It combines native decision procedures with interaction and automation but it is too expressive for easy translation into standard ATP tools.
\cite{rabe:dholmodels:26} gives a model-theoretical semantics of DHOL that includes polymorphism.
Recently the Vampire automated theorem prover has been extended with shallow
polymorphism~\cite{BR20}. The extension is very elegant, but it focuses on non-dependent types
so far.

CoqHammer~\cite{CK18} tackles a similar problem but sits on a different end of the
expressivity and automation spectrum: It translates 
Rocq into first-order logic.
It is, however, incomplete.

The general idea of translating dependent type theories is not new: \cite{FM90} experimented with a translation of
LF into hereditary Harrop formulas, a simply-typed meta-logic.
Similarly, Jacobs and Melham~\cite{jacobs_dtt_hol} give a translation of dependent type theory into higher-order logic.
Both translate dependent types to unary predicates that serve as type guards.


Such translations are not necessarily sound and complete.
This is because refinement types are not closed under function type formation, i.e., the function type on refinement types cannot be represented as a refinement of a function type.
A similar argument applies to quotients.
This motivated the use of PERs, which subsume refinements and quotients, and are closed under function type formation.
Thus, many constructions that involve refinements or quotients of base types are eventually generalized to PERs in order to complete inductive definitions or proofs.
Because dependent base types can be understood as refinements of a bigger type, the semantics of dependent types is one such construction.
PERs have been used to formulate the semantics of dependent types, in essentially the same way as we do, going back to at least \cite{persemantics} and were used for the semantics of NuPRL~\cite{nuprlpers}.
Another example are parametricity arguments in polymorphic type theories such as~\cite{BJP10}.
For example, recently, \cite{hol_pers} presented a parametricity translation from HOL to itself that interprets every type as a PER.
That translation arises as the special case of ours where the input is restricted to the HOL fragment of DHOL.
Because, their source and target language are the same, they additionally study which HOL-style subtype definitions can be translated to themselves, leading them to introduce the notion of wideness.\\[-6mm]

\paragraph*{Future Work}
Even though our automation is well-behaved, it is still not as strong as desirable.
To improve this, we plan to define dedicated automated reasoning rules for PDHOL
similar to Niederhauser et al.~\cite{dhol_satallax} for DHOL. 
Furthermore, a larger case study involving dependently typed examples from ITPs 
would allow checking if PDHOL constitutes a
useful intermediate language for proof automation.  

Finally, we want to combine our work with the extension of DHOL with subtyping from~\cite{RR:dholsub}.
While we expect simply merging the language extensions to be straightforward, their union allows adding \emph{bounded} polymorphism where type variables $\alpha<:A$ can only be instantiated with subtypes of $A$.
Intuitively, the type system and translation can be extended easily to allow such upper
bounds on type variables, but it is unclear how difficult the judgment preservation and truth
reflection proofs and the design of a (sub)type-checking algorithm will be in that case.

\subsubsection*{Acknowledgements}
Rabe was supported by the FAUstairs project (see\\ https://www.faustairs.fau.de/),
funded by the Stiftung Innovation in der Hochschul- lehre under grant StIL:1001-3096. Ranalter and
Kaliszyk were supported by the Renaissance Philanthropy grant DEEPER.

\bibliography{rabe,systems,pub_rabe,historical,extra}

\appendix
\section{Preservation}\label{app:comp}
\begin{proof}
  Thoerem~\ref{thm:preserve} is proved by straightforward induction on derivations. The
  sub-proofs for the individual proof rules follow the same structure. We therefore present
  here only one example (the \Rtypesym rule) --- applying the induction hypothesis followed by definitions of the erasure in full detail:
  \begin{align}
    &\CtxT{\tr{\Gamma}} && assumption \label{eq:1}\\
    &\tr{a:\Pg{\Delta} \type}\ \in\ \tr{T} && assumption \label{eq:2}\\
    &a:\ot{\tr\Delta} \to\type\ \in\ \tr{T} && \tr{\phantom{t}}\text{-def of DHOL Theories on }
                                             \ref{eq:2} \label{eq:3}\\
    &\per{a}:\ott{\tr{\Delta}}\to a\ \valpha\to a\ \valpha\to \bool \in\ \tr{T} &&
                                                                                 \tr{\phantom{t}}\text{-def of DHOL Theories on }\ref{eq:2} \label{eq:4}\\
    &\ass \tr{\Delta}\impl PER(\per{a}\ \valpha\ \per{\valpha}\ \vx) \in \tr{T} &&
                                                                                   \tr{\phantom{t}}\text{-def of DHOL Theories on }\ref{eq:2} \label{eq:5}\\
    \tr{\Gamma} &\vdash_{\tr{T}} \tr{\Delta} \leftarrow \tr{\delta} && assumption +
                                                                       IH \label{eq:6}\\
    \tr{\Gamma} &\vdash_{\tr{T}} \ot{\tr{\Delta}} \leftarrow \ot{\tr{\delta}} && \text{see Note in
                                                                       Def.}~\ref{def:holabbrs}
                                                                       \label{eq:7}\\
    \tr{\Gamma} &\vdash_{\tr{T}} \ott{\tr{\Delta}} \leftarrow \ott{\tr{\delta}} && \text{see Note in
                                                                       Def.}~\ref{def:holabbrs}
                                                                       \label{eq:8}\\
    \tr{\Gamma} &\vdash_{\tr{T}} a\ \ot{\tr{\delta}}:\type && \text{\Rtypesym},~\ref{eq:1},~\ref{eq:3},~\ref{eq:7} \label{eq:9}\\
    \tr{\Gamma} &\vdash_{\tr{T}} \tr{a\ \delta}:\type && \tr{\phantom{t}}\text{-def of DHOL
                                                         Types on }~\ref{eq:9} \label{eq:10}\\
    \tr{\Gamma} &\vdash_{\tr{T}} \per{a}\ \ott{\tr{\delta}} : (a\ \valpha\to a\ \valpha\to \bool)[\ott{\tr{\delta}}] &&
                                                                                                                        \text{\Rtermsym},~\ref{eq:1},~\ref{eq:4},~\ref{eq:8} \label{eq:11}\\
    \tr{\Gamma} &\vdash_{\tr{T}} \per{a}\ \ott{\tr{\delta}} : a\ \ot{\tr{\delta}}\to a\
                  \ot{\tr{\delta}}\to\bool && \text{apply substitution},~\ref{eq:11} \label{eq:12}\\
    \tr{\Gamma} &\vdash_{\tr{T}} \per{(a\ \delta)} : \tr{a\ \delta}\to \tr{a\ \delta}\to\bool
                        && \per{\phantom{t}}\text{-def and } \tr{\phantom{t}}\text{-def of DHOL
    types}~\ref{eq:12} \label{eq:13}\\
    \tr{\Gamma} &\vdash_{\tr{T}} PER(\per{a}\ \valpha\ \per{\valpha}\ \vx)[\tr{\delta}] &&
                                                                                         \text{\Rvalsym},~\ref{eq:1},~\ref{eq:5},~\ref{eq:6} \label{eq:14}\\
    \tr{\Gamma} &\vdash_{\tr{T}} PER(\per{a}\ \ott{\tr{\delta}}) && \text{apply
                                                                    substitution},~\ref{eq:14} \label{eq:15}\\
    \tr{\Gamma} &\vdash_{\tr{T}} PER(\per{(a\ \delta)}) && \per{\phantom{t}}\text{-def},~\ref{eq:15}
  \end{align}
  A remark about steps \ref{eq:7} and \ref{eq:8}: Inspection of the inference rules for
  substitutions shows, that the remark in Def.~\ref{def:holabbrs} extends to substitutions by
  just traversing the list and throwing out the corresponding elements of \(\delta\).

The same proof structure is directly applicable to the other rules, only note that \Rtypesym, \Rtermsym and \Rvalsym all proceed by
  substituting the \(\Delta\) in the premises by the \(\delta\) of the substitution.
\end{proof}

\section{Reflection}\label{app:sound}
The truth-reflection proof for PDHOL follows the one given for monomorphic DHOL by
Rothgang et al.~\cite{RRB23ext} closely, as the introduction of shallow polymorphism requires
only minor 
changes. This similarity stems from the fact that most changes to the proof rules happen to
accommodate the declarations in theory and context, while the validity rules stay the
same.  Our extension merely affects which types are possible and provides the option for polymorphic
conjectures. As reasoning can only happen on fully applied base types, there are only minor
changes to the formulation of some of the intermediate results.

In the sequel we will start with a overview over the proof idea, followed by the
necessary definitions and finally the extensions to the original proof necessitated by our
extending of the theory. 

The main challenge lies in the fact
that there are situations in which the erasure of ill-typed DHOL terms results in terms that
are well-typed in HOL. This is due to the non-injectivity of the translation: two fixed
length lists \(a:lst\: 2\) and \(b:lst\: 3\) of length 2 and 3 respectively are incomparable in
DHOL, but erasing them results in \(a:lst\) and \(b:lst\) for which equality would be well-typed.

Note that the opposite problem, namely that a well-typed DHOL term \(t\) translates to an
ill-typed HOL term \(\tr{t}\), cannot happen. This is clear due to the definition of the
translation.

As a result of this non-injectivity, a valid HOL-derivation cannot be translated into a valid
DHOL-derivation without further processing. The proof idea, then, is to show that it is
possible to transform a HOL proof of some translated, well-typed DHOL statement, into a
HOL proof that is in the well-typed image of the translation, allowing a direct translation
back into a DHOL proof. 

We proceed to show this in the following steps: First, we will show that the translation, while
not injective in general, is type-wise injective --- meaning that if \(t:A\) and \(s:A\) are
distinct DHOL terms of equal DHOL type, then so are \(\tr{t}:\tr{A}\) and \(\tr{s}:\tr{A}\).
This will allow us to associate a unique DHOL term with HOL terms that come from the 
erasure, assuming the type is known. Using this, we show that HOL proofs can be transformed
into HOL proofs that lie in the well-typed image of the translation which will finally allow us
to map said proofs to DHOL proofs of the untranslated conjecture. 

We front-load this section with the necessary definitions to highlight the relationship between
them.

\begin{definition}
Ill-typed DHOL terms \(t\) with a well-typed counterpart \(\tr{t}\)  will be called
\emph{spurious} while terms in which both --- erased and original --- terms are well-typed will be
called \emph{proper}. A \emph{improper} term \(\tr{t}\) is not in the
(translation-)image of any DHOL term \(t\).

\emph{Normalizing} an improper term results in a proper or spurious one. We introduce the
normalizing function used in the proof in Figure~\ref{fig:norm} and expand on it in the
sequel. Normalizing an already proper or spurious term returns the same term.

We extend the notion of ``proper'' from terms to contexts \(\Delta\), whenever \(\tr{\Gamma}\)
can be obtained from \(\Delta\) by adding typing assumptions. Then \(\Gamma\) is called the
\emph{quasi-preimage} of the \emph{proper} context \(\Delta\).

Furthermore, given a proper HOL context \(\Delta\), a statement \(\phi\) over this context is
called \emph{quasi-proper}, iff the normalization of \(\phi\) is \(\tr{F}\) for \(\Gamma
\vdash F : o\) and \(\Gamma\) quasi-preimage of \(\Delta\). In this case, F is called a
\emph{quasi-preimage} of \(\phi\).

As a last extension to this terminology, a validity judgment \(\Delta \vdash \phi\) is also
called \emph{proper} iff \(\Delta\) is proper and \(\phi\) is quasi-proper in this
context. Then \(\TrueT{\tr{\Gamma}}{\tr{F}}\) is called a \emph{relativization} of
\(\True{\Delta}{\phi}\) and \(\True{\Gamma}{F}\) is called a \emph{quasi-preimage} of
\(\True{\Delta}{\phi}\). 

We call an improper term \emph{almost proper} iff its normalization is not spurious. This is
equivalent to saying an improper term is almost proper iff it is quasi-proper (has a well-typed
quasi-preimage). Otherwise, it is called \emph{unnormalizably spurious}.

Finally, we give a definition of the property which allows us to translate HOL proofs into DHOL proofs. A valid HOL derivation is called \textit{admissible} iff all terms occurring in it are almost proper.
\end{definition}

\subsection{Type-wise injectivity of the translation}
Compared to the original formulation of Rothgang et al.~\cite{RRB23ext} there are some changes to
the translation. It is now possible to apply type arguments to base types and
constants. These, however, are preserved in the erasure: base types \(\tr{a\ \delta}\)
and constants \(\tr{c\ \delta}\) result in \(a\ \ot{\tr{\delta}}\) and \(c\
\ott{\tr{\delta}}\) respectively.

The other relevant change to the system is the addition of type variables \(\tp\alpha\) as an
option to the type system. These are straightforwardly translated into \(\tp\alpha\),
\(\alpha^*:\alpha\to\alpha\to\bool\), \(\ass \isper{\alpha^*}\). 
\begin{lemma}
  \label{lem:tyInj}
  Let \(s, t\) be DHOL terms of type \(A\). Assuming \(s\) and \(t\) are different, then
  \(\tr{s}:\tr{A}\) and \(\tr{t}:\tr{A}\) are different. 
\end{lemma}

\begin{proof}
  The proof proceeds by induction over the term structure. Different top-level productions
  result in different terms after the translation, so we can limit ourselves to the cases where
  both terms have the same root symbol. For non-equality terms, the only cases that need to be
  adjusted from the original proof, are where we performed changes to the translation.

  For that, notice that erasing applied base types and constants results in recursive calls
  to the translations of the applied types. By the induction hypothesis, these are already
  different, concluding the proof.

  Another interesting case occurs when the terms \(t, s\) are equalities over two types 
  erased to the same type. This is not possible for type variables \(\alpha\) as their
  translation just yields their simple counterparts.

  All remaining cases are identical to the original presentation.
\end{proof}

\subsection{Transforming HOL proofs into admissible HOL proofs}
In order to transform HOL proofs into admissible HOL proofs, the original proof defines two
functions.
First, they define a normalization function, given in Figure~\ref{fig:norm}, with changes
due to the incorporation of polymorphism highlighted. This normalization turns improper
HOL-statements into proper or spurious ones, i.e. the term \(norm[t]\) is
in the image of the translation after normalization.

\begin{figure}[t]
  \centering
  \begin{align*}
    norm[\tr{t}] &:= t \\
    norm[norm[s]] &:= norm[s] \\
    norm[A^*\ s] &:= \lambda y:\tr{A}.A^*\ s\ y \\
    norm[A^*] &:= \lambda x:\tr{A}.\lambda y:\tr{A}.A^*\ x\ y \\
    norm[\hlp{c\: \vec{A}}] &:= \hlp{c\: \vec{A}} \\
    norm[x] &:= x \\
    norm[f\ t] &:= norm[f]\ norm[t] \\
    norm[\lambda x:C.t] &:= \lambda x:C.norm[t] \\
    norm[s=_{\tr{A}} t] &:= A^*\ s\ t \\
    norm[s \Rightarrow t] &:= norm[s] \Rightarrow norm[t] \\
    norm[\forall x:\tr{A}.A^*\ x\ x\ \Rightarrow G] &:= \forall x,y:\tr{A}.A^*\ x\ y
                                                      \Rightarrow G \\
    \text{If } F \text{ not of shape } A^*\ \_\ \_ \Rightarrow \_ \text{ or } \forall x':\tr{A}.A^*\ x\ x'
    \Rightarrow \_:\\
    norm[\forall x:\tr{A}.F] &:= norm[\forall x:\tr{A}.A^*\ x\ x \Rightarrow F] \\
  \end{align*}
  \caption{\label{fig:norm} Definition of \(norm[t]\) with changes highlighted.}
\end{figure}

Second,  a \emph{normalizing statement transformation} \(sRed(t)\)
is applied to the derivation.
A normalizing statement transformation is a function that replaces terms and their
context in statements in such a way that unnormalizably spurious terms end up as almost proper
ones. In contrast to \(norm[t]\) the changes to accommodate polymorphism do not require
any changes in the definition of \(sRed(t)\), so the function (given in Figure~\ref{fig:sRed})
is identical to the one in \cite{RRB23ext}. \(sRed(t)\) proceeds by beta-eta normalizing terms and, in
case this does not make them almost proper, replaces unnormalizably spurious function
applications of type \(B\) by a ``default term'' \(\omega_B : B\) which is proper and exists
due to the non-emptiness assumption in HOL.

\begin{figure}[t]
  \centering
  \begin{align*}
    sRed(t_A) &:= t_A \\
    & \qquad \text{if } t \text{ has quasi-preimage of type A} \\
    sRed(f_{\Pi_{x:A}B}\ t_A) &:= sRed(sRed(f_{\Pi_{x:A}B})\ sRed(t_A)) \\
    & \qquad \text{if } f_{\Pi x:A.B}\ t_A \text{ not beta-eta reducible} \\
    \text{In the following, the term } t_A \text{ on the left-hand } & \text{ side is assumed to not be} &&\\
    \text{ almost proper with a quasi-preimage of type } &A:\\
    sRed(t_A) &:= sRed(t_A^{\beta\eta}) \\
    & \qquad \text{if } t \text{ eta-beta reducible} \\
    sRed(s_A =_{\tr{A}} t_{A'}) &:= sRed(s_A) =_{\tr{A}} sRed(t_A) \\
    sRed(F_o \Rightarrow G_o) &:= sRed(F_o) \Rightarrow sRed(G_o) \\
    sRed(\lambda x:A.s_B) &:= \lambda x:A.sRed(s_B) \\
    sRed((sRed(f_{\Pi_{x:A}B})_{\Pi_{x:A}B}\ sRed(t_{A'})_{A'})_{B'}) &:= \omega_{\tr{B}} \\
    & \qquad \text{ if } A \neq A' \text{ or } B \neq B'
  \end{align*}
  \caption{Definition of \(sRed(t)\).}
  \label{fig:sRed}
\end{figure}

This is aided by passing, for each term, a DHOL type \(A\) to the function, effectively associating them with a
quasi-preimage. This is mainly necessary for \(\lambda\)-functions where there are potentially
many quasi-preimages of differing types. To ensure correctness, it is, of course, required that an
indexed term \(t_A\) is of type \(\tr{A}\) and, if it is almost proper with a unique
quasi-preimage, that the quasi-preimage has type \(A\). 

Using the definition of the normalizing statement transformation, we can go on and state
\begin{lemma}
  \label{lem:sRed}
Assume a well-typed DHOL theory \(T\) and a conjecture \(\TrueG{\phi}\) with \(\Gamma\)
well-formed and \(\phi\) well-typed. Assume a valid HOL derivation of \(\TrueT{\tr{\Gamma}}
{\tr{\phi}}\). Then, we can index the terms in the derivations s.t. any steps S
in the derivation can be replaced by a macro-step (i.e. a step with the same assumptions and conclusion as the original step, composed of multiple micro-steps) for the normalizing statement transformation,
replacing step S s.t. after replacing all steps with their macro-steps:
\begin{itemize}
\item the resulting derivation is valid,
\item all terms occurring in the derivation are almost proper
\end{itemize}
\end{lemma}

\begin{remark}[A note about indices]
  It is reasonable to assume that the addition of type variables changes the indexing
  procedure, so we give the full (original) procedure here:

  Indexing starts at the end of the derivation. We pick identical indices for identical terms.
  Whenever we need to index a constant or variable, and its preimage exists in the context of
  the DHOL conjecture, we pick the type of the preimage. Term equalities always have the same
  index on both sides. For non-atomic terms \(t\), we pick indices for the atomic subterms and
  choose for \(t\) a type of the unique (by Lemma~\ref{lem:tyInj}) quasi-preimage, such that
  the indices match up. If possible, we choose indices such that there exists a well-typed
  quasi-preimage of that type. Unless otherwise indexed, the index for \(sRed(t_A)\) will also
  be \(A\). For \(\lambda\)-functions that already have an index assigned, the variable and
  the body are assigned matching indices. If the \(\lambda\) function does not yet have an
  index, but is applied to an argument with an index, we assign that index to the variable
  of the \(\lambda\)-function.

  If none of these rules apply, we pick an arbitrary index that does not violate any of the
  future applications of the rules.

  Inspecting this algorithm, it becomes clear that it forbids at no point the choice of a type
  variable. Indeed, the labeling process is completely agnostic to the underlying set of
  types. 
\end{remark}

Due to the fact that no changes to the definition of \(sRed\) are necessary, and the conjecture
only talks about validity statements, we refer to the original proof in \cite{RRB23ext} for
details. We will nevertheless give one example derivation to illustrate how the function interacts with the labels:

\begin{proof}
 The proof proceeds by induction on the inference rules, and we pick the beta rule as
example: \[\rul[rbeta]{beta}{\Gamma \vdash_T (\lambda x:A.s)\ t:B}{\Gamma\vdash_T (\lambda x:A.s)\ t=_B s[x/t]}\]

For the sake of clarity, we will use the substitution notation used by Rothgang et
al. in~\cite{RRB23ext}. Here \(t[x/u]\) stands for the capture avoiding substitution of the
term \(x\) with the term \(u\) in \(t\).

By assumption we get \[\Delta \vdash_{\tr{T}} sRed((\lambda x_A:\tr{A}.s_B)\ t_A)_{B'}:\tr{B}.\]

We proceed by case distinction on whether it is almost proper with quasi-preimage of type \(B
\equiv B'\):

If it is, we get \[\Delta \vdash_{\tr{T}} ((\lambda x_A:\tr{A}.s_B)\ t_A)_{B'}:\tr{B}\] by
the first case in the definition of \(sRed\). We apply the beta rule and the definition
of \(sRed\)'s first case twice and the equality case once to that result to get the
goal \[\Delta \vdash_{\tr{T}} sRed ((\lambda x_A:\tr{A}.s_B)\ t_A =_{\tr{B}} s_B[x_A/t_A]).\]

If not, we observe that \[sRed((\lambda x_A:\tr{A}.s_B)\ t_A) = sRed(sRed(s_B)[x_A/sRed(t_A)]) =
sRed(s_B)[x_A/sRed(t_A)].\] From reflexivity we have \[\Delta \vdash_{\tr{T}}
sRed(s_B)[x_A/sRed(t_A)] =_{\tr{B}} sRed(s_B)[x_A/sRed(t_A)].\] Due to the induction hypothesis
and the choice of indices, we can assume that the \(sRed\) terms in the equality have a
quasi-preimage of type \(B\), and by several applications of the definition of \(sRed\) we
conclude \[\Delta \vdash_{\tr{T}} sRed ((\lambda x_A:\tr{A}.s_B)\ t_A =_{\tr{B}} s_B[x_A/t_A]).\]

Note how this transformed a single beta rule step into a macro-step. The derivation
shows that even if during a regular proof step the statement would become unnormalizably
spurious, we can transform the statements in a way that yields almost proper terms.
\end{proof}

\subsection{Translation of HOL proofs into DHOL proofs}
Finally, we show the reflection of truth theorem. As previously, we give the
general outline of the proof and refer to \cite{RRB23ext} for details.
\begin{proof}
  According to Lemma~\ref{lem:sRed} we can assume that the proof of \(\TrueT{\tr{\Gamma}}
  {\tr{F}}\) is admissible, as we can always transform a valid HOL proof into an
  admissible and valid HOL proof. Because admissibility implies the existence of a well-typed
  quasi-preimage and the fact that the translation is type-wise injective, we therefore
  have that the translated conjecture is a proper validity statement with unique quasi-preimage
  in DHOL.

  It remains to show that it is possible to lift the HOL derivation of the conjecture to a DHOL
  derivation of its quasi-preimage. For that, we can inspect the validity rules one for one and
  show that --- assuming the conclusion is proper and has a quasi-preimage --- all validity
  assumptions and their contexts are well-formed and proper respectively. From this, we
  continue to prove that in this case, the quasi-preimage of the conclusion of the rule is valid.

  As stated previously, the validity rules for the polymorphic extension do not change compared
  to their monomorphic variants. However, we now have to consider applied types. Inspecting the
  normalization function shows that normalizing constants applied to type arguments is the identity
  function. Therefore, there is no change in the validity of the proofs as performed in
  \cite{RRB23ext}.
\end{proof}

\section{Red-Black Tree in PDHOL}\label{app:rbt}
Following is the theory of red-black trees from our case study in PDHOL:
\begin{gather*}
  \clr:\type\quad\bl:\clr\quad\rd:\clr\\
  \nat:\type\quad 0:\nat\quad \suc:\nat\to\nat\\
  \tree:\hlp{\Pt{\alpha}}\hlm{\Pi_{c:\clr,n:\nat}} \type
  \quad\leaf:\hlp{\Pt{\alpha}} \tree\ \alpha\ \bl\ 0\\
  \rt:\hlp{\Pt{\alpha}}\hlm{\P{n}{\nat}}\tree\ \alpha\ \bl\ n \to \alpha \to \tree\ \alpha\
  \bl\ n \to \tree\ \alpha\ \rd\ n\\
  \bt:\hlp{\Pt{\alpha}}\hlm{\Pi_{c_1:\clr,c_2:\clr,n:\nat}} \tree\ \alpha\ c_1\ n \to
  \alpha \to \tree\ \alpha\  c_2\ n \to\\ \quad \tree\ \alpha\ \bl\ (\suc\ n)\\
  \rev:\hlp{\Pt{\alpha}}\hlm{\Pi_{c:\clr,n:\nat}} \tree\ \alpha\ c\ n \to \tree\ \alpha\ c\
  n\\
  \ass\forall\alpha:\type.\rev\ \alpha\ \bl\ 0\ (\leaf\ \alpha) = \leaf\ \alpha\\
  \ass\forall\alpha:\type,x:\alpha,n:\nat,c_1:\clr,c_2:\clr,t_1:\tree\ \alpha\ c_1\
  n,t_2:\tree\ \alpha\ c_2\ n.\\
  \quad \rev\ \alpha\ \bl\ (\suc\ n)\ (\bt\ \alpha\ n\ c_1\ c_2\ t_1\ x\ t_2) = \\
  \quad \bt\ \alpha\ n\ c_2\ c_1\ (\rev\ \alpha\ c_2\ n\ t_2)\ x\ (\rev\ \alpha\ c_1\ n\ t_1)\\
  \ass\forall\alpha:\type,x:\alpha,n:\nat,t_1:\tree\ \alpha\ \bl\ n,t_2:\tree\ \alpha\ \bl\ n.\\
  \quad \rev\ \alpha\ \rd\ n\ (\rt\ \alpha\ n\ t_1\ x\ t_2) =\\
  \quad \rt\ \alpha\ n\ (\rev\ \alpha\ \bl\ n\ t_2)\ x\ (\rev\ \alpha\ \bl\ n\ t_1)\\
\end{gather*}
From this various in-between steps, e.g. base cases and instanced induction schemes, were
shown. In the interest of readability these results, as well as the induction axiom, are
omitted here. Taking everything together, this allows us to proof the following involution
lemma:  
\begin{gather*}
  \ass\forall\alpha:\type,n:\nat,c:\clr,t:\tree\ \alpha\ c\ n.\\
  \quad \rev\ \alpha\ c\ n\ (\rev\ \alpha\ c\ n\ t) = t
\end{gather*}

\section{Experimental Results}
\begin{table}
\begin{tabular}{l|r|r|r||l|r|r|r}
 & \multicolumn{3}{c}{Time to prove} & & \multicolumn{3}{c}{Time to prove}\\
  File name & Vampire & Leo & Zipper & File name & Vampire & Leo & Zipper \\
  \hline
  \texttt{rbt-rev-invol} & 13.48 & TO & TO & \texttt{list-app-assoc} & 13.37 & TO & TO \\
  \quad \rotatebox[origin=c]{180}{$\Lsh$}\texttt{TCO} & trivial & trivial & trivial & \quad
                                                                                      \rotatebox[origin=c]{180}{$\Lsh$}\texttt{TCO}
                    & 0.02 & 0.70 & 0.35 \\
  \texttt{-base} & 33.01 & TO & TO & \texttt{-TI} & TO & TO & TO \\
  \quad \rotatebox[origin=c]{180}{$\Lsh$}\texttt{TCO} & trivial & trivial & trivial & \quad
                                                                                      \rotatebox[origin=c]{180}{$\Lsh$}\texttt{TCO}
                    & 0.02 & 0.66 & 0.38\\
  \texttt{-base-BTLeaf} & 0.08 & TO & TO & \texttt{-TI\_alt} & 14.66 & TO & TO \\
  \quad \rotatebox[origin=c]{180}{$\Lsh$}\texttt{TCO} & trivial & trivial & trivial & \quad
                                                                                      \rotatebox[origin=c]{180}{$\Lsh$}\texttt{TCO}
                    & 0.02 & 0.81 & 0.36 \\
  \texttt{-base-RTLeaf} & 8.03 & TO & TO & \texttt{-base} & 5.18 & TO & TO \\
  \quad \rotatebox[origin=c]{180}{$\Lsh$}\texttt{TCO} & trivial & trivial & trivial & \quad
                                                                                      \rotatebox[origin=c]{180}{$\Lsh$}\texttt{TCO}
                    & 0.02 & 0.65 & <0.01 \\
  \texttt{-base-lemma} & 18.01 & TO & TO & \texttt{-base-TI} & 10.83 & TO & TO \\
  \quad \rotatebox[origin=c]{180}{$\Lsh$}\texttt{1.TCO} & 0.02 & 0.64 & <0.01 & \quad
                                                                                 \rotatebox[origin=c]{180}{$\Lsh$}\texttt{TCO}
                    & 0.02 & 0.47 & <0.01\\
  \quad \rotatebox[origin=c]{180}{$\Lsh$}\texttt{2.TCO} & 0.02 & 0.67 & <0.01 &
                                                                                 \texttt{-indinst}
                    & TO & TO & TO \\
  \texttt{-subBTStep1} & 0.46 & 2.99 & TO & \quad
                                          \rotatebox[origin=c]{180}{$\Lsh$}\texttt{1.TCO} &
                                                                                             0.02
                          & 0.54 & <0.01 \\
  \quad \rotatebox[origin=c]{180}{$\Lsh$}\texttt{TCO} & trivial & trivial & trivial & \quad
                                                                                      \rotatebox[origin=c]{180}{$\Lsh$}\texttt{2.TCO}
                    & 0.02 & 0.72 & 1.81 \\
  \texttt{-subBTStep2} & 0.22 & 47.48 & TO & \quad
                                          \rotatebox[origin=c]{180}{$\Lsh$}\texttt{3.TCO} &
                                                                                             0.02
                          & 0.90 & 1.29 \\
  \quad \rotatebox[origin=c]{180}{$\Lsh$}\texttt{TCO} & trivial & trivial & trivial & \quad
                                                                                      \rotatebox[origin=c]{180}{$\Lsh$}\texttt{4.TCO}
                    & 0.02 & 0.72 & 1.71 \\
  \texttt{-subBTStep3} & 0.08 & 2.80 & 2.32 & \texttt{-indinst-TI} & TO & TO & TO \\
  \quad \rotatebox[origin=c]{180}{$\Lsh$}\texttt{TCO} & trivial & trivial & trivial & \quad
                                                                                      \rotatebox[origin=c]{180}{$\Lsh$}\texttt{1.TCO}
                    & 0.02 & 0.56 & <0.01 \\
  \texttt{-composedBTStep} & 10.62 & TO & TO & \quad
                                               \rotatebox[origin=c]{180}{$\Lsh$}\texttt{2.TCO}
                    & 0.02 & 0.80 & 1.60\\
  \quad \rotatebox[origin=c]{180}{$\Lsh$}\texttt{TCO} & trivial & trivial & trivial & \quad
                                                                                      \rotatebox[origin=c]{180}{$\Lsh$}\texttt{3.TCO}
                    & 0.02 & 0.78 & 1.34 \\
  \texttt{-BTStep} & 35.49 & TO & TO & \quad \rotatebox[origin=c]{180}{$\Lsh$}\texttt{4.TCO} &
                                                                                                0.02
                          & 0.76 & 1.70 \\
  \quad \rotatebox[origin=c]{180}{$\Lsh$}\texttt{TCO} & trivial & trivial & trivial &
                                                                                      \texttt{-indstep1}
                    & 10.56 & TO & TO \\
  \texttt{-RTStep} & 5.67 & TO & TO & \quad \rotatebox[origin=c]{180}{$\Lsh$}\texttt{1.TCO} &
                                                                                               0.02
                          & 0.56 & 0.37 \\
  \quad \rotatebox[origin=c]{180}{$\Lsh$}\texttt{TCO} & trivial & trivial & trivial & \quad
                                                                                      \rotatebox[origin=c]{180}{$\Lsh$}\texttt{2.TCO}
                    & 0.02 & 1.93 & 0.69\\
  \texttt{-indinst} & TO & TO & TO & \texttt{-indstep1-TI} & TO & TO & TO \\
  \quad \rotatebox[origin=c]{180}{$\Lsh$}\texttt{TCO} & trivial & trivial & trivial & \quad
                                                                                      \rotatebox[origin=c]{180}{$\Lsh$}\texttt{1.TCO}
                    & 0.02 & 0.63 & 0.50 \\
  \texttt{bad\_tree} & TO & TO & TO & \quad \rotatebox[origin=c]{180}{$\Lsh$}\texttt{2.TCO} &
                                                                                               0.02
                          & 1.43 & 0.49 \\
  \quad \rotatebox[origin=c]{180}{$\Lsh$}\texttt{1.TCO} & TO & TO & TO & \texttt{-indstep2} &
                                                                                               10.58
                          & TO & TO \\
  \quad \rotatebox[origin=c]{180}{$\Lsh$}\texttt{2.TCO} & TO & TO & TO & \quad
                                                                          \rotatebox[origin=c]{180}{$\Lsh$}\texttt{1.TCO}
                    & 0.02 & 0.64 & 0.40 \\
  \texttt{list-app-nil} & 1.39 & TO & TO & \quad \rotatebox[origin=c]{180}{$\Lsh$}\texttt{2.TCO}
                    & 0.02 & 0.67 & 0.68\\
  \quad \rotatebox[origin=c]{180}{$\Lsh$}\texttt{TCO} & 0.02 & 0.65 & <0.01 &
                                                                              \texttt{indstep2-TI}
                    & TO & TO & TO \\
  \texttt{-base} & 0.02 & TO & 2.18 & \quad \rotatebox[origin=c]{180}{$\Lsh$}\texttt{1.TCO} &
                                                                                              0.02
                          & 0.64 & 0.39 \\
  \quad \rotatebox[origin=c]{180}{$\Lsh$}\texttt{TCO} & 0.02 & 0.63 & <0.01 & \quad
                                                                              \rotatebox[origin=c]{180}{$\Lsh$}\texttt{2.TCO}
                    & 0.02 & 1.20 & 0.39 \\
  \texttt{-indinst} & TO & TO & TO & \texttt{-indstep3} & 0.4 & TO & TO \\
  \quad \rotatebox[origin=c]{180}{$\Lsh$}\texttt{1.TCO} & 0.02 & 0.60 & <0.01 & \quad
                                                                                \rotatebox[origin=c]{180}{$\Lsh$}\texttt{1.TCO}
                    & 0.02 & 0.65 & 0.33 \\
  \quad \rotatebox[origin=c]{180}{$\Lsh$}\texttt{2.TCO} & 0.02 & 0.64 & <0.01 & \quad
                                                                                \rotatebox[origin=c]{180}{$\Lsh$}\texttt{2.TCO}
                    & 0.02 & 8.60 & 0.09 \\
  \quad \rotatebox[origin=c]{180}{$\Lsh$}\texttt{3.TCO} & 0.02 & 0.67 & <0.01 & \quad
                                                                                \rotatebox[origin=c]{180}{$\Lsh$}\texttt{3.TCO}
                    & 0.02 & 8.65 & 0.61 \\
  \quad \rotatebox[origin=c]{180}{$\Lsh$}\texttt{4.TCO} & 0.02 & 0.72 & <0.01 &
                                                                                \texttt{indstep3-TI}
                    & 6.72 & TO & TO \\
  \texttt{-indstep} & 5.27 & TO & TO & \quad \rotatebox[origin=c]{180}{$\Lsh$}\texttt{1.TCO} &
                                                                                               0.02
                          & 0.62 & 0.35 \\
  \quad \rotatebox[origin=c]{180}{$\Lsh$}\texttt{1.TCO} & 0.02 & 0.55 & <0.01 & \quad
                                                                                \rotatebox[origin=c]{180}{$\Lsh$}\texttt{2.TCO}
                    & 0.02 & 1.92 & 0.02 \\
  \quad \rotatebox[origin=c]{180}{$\Lsh$}\texttt{2.TCO} & 0.02 & 0.52 & <0.01 &  \quad
                                                                                \rotatebox[origin=c]{180}{$\Lsh$}\texttt{3.TCO}
                    & 0.02 & 1.41 & 0.40 \\
  \texttt{list-rev-invol} & 11.15 & TO & TO & \texttt{-indstep4} & 1.43 & 2.69 & TO \\
  \quad \rotatebox[origin=c]{180}{$\Lsh$}\texttt{TCO} & trivial & trivial & trivial & \quad
                                                                                      \rotatebox[origin=c]{180}{$\Lsh$}\texttt{TCO}
                    & 0.01 & 0.58 & 0.01 \\
  \texttt{-TI} & 8.55 & TO & TO & \texttt{-indstep4-TI} & 6.72 & TO & TO \\
  \quad \rotatebox[origin=c]{180}{$\Lsh$}\texttt{TCO} & trivial & trivial & trivial & \quad
                                                                                      \rotatebox[origin=c]{180}{$\Lsh$}\texttt{TCO}
                    & 0.02 & 0.53 & 0.02 \\
  \texttt{-step1} & 0.09 & TO & TO & \texttt{-indstep5} & TO & TO & TO \\
  \quad \rotatebox[origin=c]{180}{$\Lsh$}\texttt{TCO} & 0.02 & 0.82 & 0.25 & \quad
                                                                             \rotatebox[origin=c]{180}{$\Lsh$}\texttt{1.TCO}
                    & 0.02 & 0.70 & 0.35 \\
  \texttt{-step1-TI} & 6.60 & TO & TO & \quad \rotatebox[origin=c]{180}{$\Lsh$}\texttt{2.TCO} &
                          0.03 & 1.22 & 0.74 \\
  \quad \rotatebox[origin=c]{180}{$\Lsh$}\texttt{TCO} & 0.02 & 0.75 & 0.19 & & & & \\
  \texttt{-step2} & 5.40 & TO & TO & & & & \\
  \quad \rotatebox[origin=c]{180}{$\Lsh$}\texttt{1.TCO} & 0.02 & 0.56 & 0.35 & & & & \\
  \quad \rotatebox[origin=c]{180}{$\Lsh$}\texttt{2.TCO} & 0.02 & 2.26 & <0.01 & & & & \\
\end{tabular}
\end{table}
\begin{table}[t]
\begin{tabular}{l|r|r|r||l|r|r|r}
  \texttt{-step2-TI} & 6.63 & TO & TO & \quad \texttt{-indstep5-TI} & TO & TO & TO \\
  \quad \rotatebox[origin=c]{180}{$\Lsh$}\texttt{1.TCO} & 0.02 & 0.80 & 0.29 & \quad \rotatebox[origin=c]{180}{$\Lsh$}\texttt{1.TCO} & 0.02 & 0.71 & 0.35 \\
  \quad \rotatebox[origin=c]{180}{$\Lsh$}\texttt{2.TCO} & 0.01 & 0.61 & <0.01 & \quad
                                                                                \rotatebox[origin=c]{180}{$\Lsh$}\texttt{2.TCO}
                                                                    & 0.02 & 1.45 & 0.41 \\
  \texttt{-step3} & 5.31 & TO & TO & \texttt{indstep5-TI\_alt} & 0.18 & TO & TO \\
  \quad \rotatebox[origin=c]{180}{$\Lsh$}\texttt{1.TCO} & 0.02 & 0.49 & <0.01 & \quad \rotatebox[origin=c]{180}{$\Lsh$}\texttt{1.TCO} & 0.02
                          & 0.63 & 0.48 \\
  \quad \rotatebox[origin=c]{180}{$\Lsh$}\texttt{2.TCO} & 0.02 & 0.37 & <0.01 & \quad
                                                                               \rotatebox[origin=c]{180}{$\Lsh$}\texttt{2.TCO}
                    & 0.02 & 0.86 & 0.42 \\
  \texttt{-step3-TI} & 6.60 & TO & TO  &
                                                                                \texttt{inst-poly-matrix-add-com}
                    & TO & TO & TO \\ 
  \quad \rotatebox[origin=c]{180}{$\Lsh$}\texttt{1.TCO} & 0.02 & 0.56 & <0.01 & \quad \rotatebox[origin=c]{180}{$\Lsh$}\texttt{TCO} &
                                                                                           trivial
                          & trivial & trivial \\
  \quad \rotatebox[origin=c]{180}{$\Lsh$}\texttt{2.TCO} & 0.02 & 0.60 & <0.01 & \texttt{inst-poly-matrix-mul-com}
                    & TO & TO & TO \\
  \texttt{-step4} & 0.18 & TO & TO & \quad
                                                                                \rotatebox[origin=c]{180}{$\Lsh$}\texttt{TCO}
                    & trivial & trivial & trivial \\
  \quad \rotatebox[origin=c]{180}{$\Lsh$}\texttt{TCO} & 0.02 & 0.56 & <0.01 &
                                                                              \texttt{finite-type}
                                                                    & TO & CA & TO \\
  \texttt{-step4-TI} & 6.60 & TO & TO & & & \\
  \quad \rotatebox[origin=c]{180}{$\Lsh$}\texttt{TCO} & 0.02 & 0.56 & <0.01 & & & \\
  \texttt{-step5} & 0.02 & TO & 1.12 & & & \\
  \quad \rotatebox[origin=c]{180}{$\Lsh$}\texttt{TCO} & trivial & trivial & trivial & & & \\
  \texttt{-step5-TI} & 6.62 & TO & TO & & & \\
  \quad \rotatebox[origin=c]{180}{$\Lsh$}\texttt{TCO} & trivial & trivial & trivial & & & \\
  \texttt{-step5-TI\_alt} & 0.03 & TO & 1.27 & & & \\
  \quad \rotatebox[origin=c]{180}{$\Lsh$}\texttt{TCO} & trivial & trivial & trivial & & & \\
\end{tabular}
\end{table}
 TCOs with the entry
``trivial'' were equalities that could be dismissed by reflexivity. ``TO'' represents a
timeout, CA the SZS Status ContradictoryAxioms. \texttt{-TI} stands for \textit{T}ype
\textit{I}nstantiated and represent our 
experiments where type arguments are applied to polymorphic axioms and conjectures.

\label{app:res}
\end{document}